\documentclass{article}
\newtheorem{theorem}{Theorem}[section]
\newtheorem{lemma}{Lemma}[section]
\newtheorem{proposition}{Proposition}[section]
\newtheorem{corollary}{Corollary}[section]
\newenvironment{proof}{\textbf{Proof.}}{$\Box$}
\newtheorem{definition}{Definition}[section]
\newtheorem{example}{Example}[section]

\newtheorem{conjecture}{Conjecture}[section]
\newtheorem{problem}{Problem}[section]

\usepackage{lineno}
\usepackage[all]{xy}
\usepackage{natbib}
\usepackage[english]{babel}
\usepackage[utf8]{inputenc}
\usepackage[T1]{fontenc}
\usepackage{amssymb, amsmath} 
\usepackage{graphicx}
\usepackage{color} 
\usepackage{hyperref}
\usepackage[some]{background}
\usepackage{xcolor}
\usepackage{eurosym}
\usepackage{longtable}
\usepackage{textcomp}
\usepackage{fancybox}
\usepackage{caption}
\usepackage{framed} 
\usepackage{multicol}
\usepackage{multirow}
\usepackage[strict]{changepage}
\usepackage{mathrsfs}
\usepackage{numprint}
\usepackage{helvet}
\usepackage{multirow} 
\newtheorem{algo}{Algorithm}[section]
\usepackage{amsmath}
\usepackage{mathtools}
\usepackage[a4paper]{geometry}
\usepackage{subcaption}
\def\cP{{\mathcal{P}}}

\def\CC{{\mathbb{C}}}
\def\NN{{\mathbb{N}}}

\def\RR{{\mathbb{R}}}

\title{Determinantal sampling designs}
\author{V. Loonis\footnote{Insee, Division des M\'ethodes et 
des R\'ef\'erentiels G\'eographiques, Paris, France. Email:vincent.loonis@insee.fr}, X. Mary\footnote{Modal'X, UPL, Univ Paris Nanterre, F92000 Nanterre France. Email:xavier.mary@u-paris10.fr}}

\date{}

\begin{document}
\maketitle

\begin{abstract}
In this article, recent results about point processes are used in sampling theory. Precisely, we define and study a new class of sampling designs: determinantal sampling designs. The law of such designs is known, and there exists a simple selection algorithm. We compute exactly the variance of linear estimators constructed upon these designs by using the first and second order inclusion probabilities. Moreover, we obtain asymptotic and finite sample theorems. We construct explicitly fixed size determinantal sampling designs with given first order inclusion probabilities. We also address the search of optimal determinantal sampling designs. 
\end{abstract}



\section{Introduction}
The goal of sampling theory is to acquire knowledge of a parameter of interest $\theta$ using only partial information. The parameter $\theta$ is a function of $\{y_{k}, k\in U\}$, usually the sum or the mean of the $y_k$'s. This is done by means of a sampling design, through which a random subset $\{y_k,k\in U\}$ is observed, and the construction of an estimator $\hat{\theta}$ of $\theta$ based on this random sample. The properties of the sampling design are thus of crucial importance to get ``good'' estimators. In practice, the following issues are fundamental:
\begin{itemize}
\item simplicity of the design,
\item knowledge of the first and, possibly, second order inclusion probabilities, 
\item control of the size of the sample,
\item effective construction, in particular with prescribed unequal probabilities,
\item statistical amenability (consistency, central limit theorem,...),
\item low Mean Square Error (MSE)/Variance of the estimator. 
\end{itemize}

In this article, we introduce a new parametric family of sampling designs indexed by Hermitian contracting matrices, \textit{determinantal sampling designs}, that addresses all theses issues. Section \ref{secDet} gives their definition and probabilistic properties. In particular, it is shown that for this family, inclusion probabilities are known for any order. Section \ref{secDet} also provides a selection algorithm. Section \ref{secEst} studies the statistical properties of linear estimators of a total. It gives algebraic and geometric formulas for the $\mathrm{MSE}$ which provide necessary and sufficient conditions for obtaining a perfectly balanced determinantal sampling design. In addition, we give asymptotic theorems and concentration inequalities. Sections \ref{Sec:ConstructingEqual} and \ref{Sec:ConstructingUnequal} provide effective constructions of fixed size determinantal sampling designs with fixed first order inclusion probabilities. Optimal properties are discussed in Section \ref{SecOpt}. Finally Section \ref{SecSim} shows simulation studies. In particular, we explore the empirical optimality of the sampling design based upon the matrix constructed in Section \ref{Sec:ConstructingUnequal}.



\section{Definition and general properties}\label{secDet}

\subsection{Definition}
According to its definition, an \textit{unordered sampling design without replacement} (simply called \textit{sampling design} afterwards) is a \textit{simple point process} on a finite set $U$, that is to say a probability on $2^U$, set of parts of $U$ (\citet{borodin2009determinantal}, \citet{tille2011sampling}). 

Among simple point processes, the general structure and properties of \textit{determinantal point processes} have attracted a lot of attention recently (\citet{borodin2009determinantal}, \citet{hough2006determinantal}, \citet{hough2009zeros}, \citet{lyons2003determinantal}, \citet{soshnikov2000determinantal}). This is (in part) due to the ubiquity of determinantal point processes in probability theory. They appear for instance in the study of random structures such as uniform spanning trees, zeros of random polynomials and spectra of random matrices. In the case of a finite set $U$, determinantal point processes are defined through associated matrices called kernels (such a kernel is however not unique). Many probabilistic properties of these processes therefore depend on algebraic properties of their kernels, but most of the results concern Hermitian matrices only. For this reason, and though there exist many interesting examples of determinantal point processes associated to non-Hermitian matrices, \textbf{we restrict our attention to the Hermitian case}. 

Unless specifically stated, matrices will be complex matrices. For a complex number $z$, $\overline{z}$ is its conjugate and $|z|=\sqrt{z\overline{z}}$ its modulus. We introduce the following notation. For any square matrix $K$ indexed by $U$ and $s\subseteq U$, $K_{|s}$ denotes the submatrix of $K$ whose rows and columns are indexed by $s$. We will also use the following convention: the determinant of the empty matrix is $1$, as is a product over the empty set $(\prod_{k \in \emptyset} \alpha_k=1)$. From the definition of determinantal point processes we derive the following definition of \textit{determinantal sampling designs}:
 
\begin{definition}[Determinantal sampling design]\label{defDSD}
A sampling design $\mathcal{P}$ on a finite set $U$ is a determinantal sampling design if there exists a Hermitian matrix $K$ indexed by $U$, called kernel, such that for all $s \in 2^U$, $ \sum_{s'\supseteq s} \cP(s')=\det(K_{|s}).$ This sampling design is denoted by $DSD(K)$.\\
A random variable $\mathbb{S}$ with values in $2^U$ and law $DSD(K)$ is called a determinantal random sample (with kernel $K$). It satisfies, for all $s \in 2^U$,
\[pr(s \subseteq \mathbb{S})=\det(K_{|s}).\]
We will also write $\mathbb{S}\sim DSD(K)$.
\end{definition}


In the following we will always identify the finite population $U$ of size $N$ with $\{1,\ldots,N\}$. It follows from the definition that determinantal sampling designs are unordered and without replacement. \citet{macchi1975coincidence} and \citet{soshnikov2000determinantal} proved that a Hermitian matrix $K$ defines a determinantal point process, and as a consequence a $DSD(K)$, iff (if and only if) $K$ is a \textit{contracting matrix}, that is a matrix whose eigenvalues are in $\left[0,1\right]$. It follows from this fundamental result that determinantal sampling designs form a parametric family of sampling designs, parametrized by contracting matrices. 

\begin{example}[Poisson sampling]\label{exPoisson}
Consider a diagonal matrix $K^{\Pi}$ with diagonal elements $K^{\Pi}_{kk}=\Pi_k$ with values in $[0,1]$. It is a contracting matrix and the corresponding determinantal sampling design satisfies, for all $s\in 2^U$, 
\begin{equation*}
pr(s \subseteq \mathbb{S} )=\underset{k \in s}\prod \Pi_k.\end{equation*} 
The inclusion-exclusion principle implies that
\begin{equation*}
pr(\mathbb{S}=s)=\underset{k \in s}\prod \Pi_k\underset{k \notin s}\prod (1-\Pi_k) \label{poisson}.\end{equation*} This is precisely the equation of the Poisson sampling design (with first order inclusion probabilities $pr(k\in \mathbb{S})=\Pi_k$), which therefore belongs to the family of determinantal sampling designs.
\end{example}

Let $K$ be a Hermitian projection matrix. Then $K=\overline{K}\,^T$ and $K^2=K$, hence $K$ is an orthogonal projection matrix. Therefore, we will make no distinction between projections and orthogonal projections. As the eigenvalues of $K$ are $0$ or $1$, then $K$ is a contracting matrix. We can thus associate to $K$ a determinantal sampling design $DSD(K)$. We will see that $DSD(K)$ enjoys interesting statistical and computational properties. Such determinantal point processes are sometimes called determinantal projection processes (\citet{hough2006determinantal}) or elementary determinantal point processes (\citet{kulesza2011learning}) in the literature. \\
We will usually write the projection matrix $K$ as $K=V\overline{V}\,^T$, where $V$ is the $(N\times n)$ matrix of an $n$ orthonormal basis of the range of $K$.

Among these sampling designs,we single out three particular cases.
\begin{example}[Projection]\label{exProj}
Let $J_N$ be the square matrix of size $N$ with all terms equal to $1$.
\begin{enumerate}
\item $DSD(\frac{1}{N}J_N)$ is the simple random sampling (SRS) of size $1$.
\item $DSD(I_N- \frac{1}{N}J_N)$ is the SRS of size $N-1$. 
\item If $K$ is a diagonal projection matrix, $DSD(K)$ is a degenerated (non-random) sampling design.
\end{enumerate} 
\end{example}

Apart from the cases $n=N-1$ and $n=1$, \citet{kulesza2012learning} proved that the SRS is not a determinantal sampling design.

\subsection{Inclusion probabilities}

The following formulas for the inclusion probabilities of order $1$ and $2$ follow from Definition \ref{defDSD}. As usual in sampling theory, we denote them by $\pi_k$ and $\pi_{kl}$. In matrix formulation, for all $k,l\in U,$ setting 

\begin{eqnarray}
\pi_k &=& pr(k\in \mathbb{S})= K_{kk}, \label{pikd}\\
\pi_{kl} &=& pr(k,l\in \mathbb{S})=K_{kk}K_{ll}-\mid K_{kl} \mid ^2 \, (k\neq l), \label{pikld}
\\ 
\label{deltakld}
\Delta_{kl} &=& \left \{ \begin{array}{l} 
\pi_{kl}-\pi_k\pi_l=-\mid K_{kl} \mid ^2 \,(k \neq l), \\
\pi_k(1-\pi_k)=K_{kk}(1-K_{kk}) \,(k = l).\\
\end{array} \right. 
\end{eqnarray} 
it holds that
\begin{equation}\label{deltad}
\Delta =\overline{(I_N-K)}*K= (I_N-K)*\overline{K}, \end{equation} 
where $*$ is the Schur-Hadamard (entrywise) matrix product.

\begin{proposition}\label{senyates}
From \eqref{deltakld} a determinantal sampling design satisfies the so-called \textit{Sen-Yates-Grundy conditions}: 
\begin{equation}\label{EqSYG}
\pi_{kl}\leq \pi_k\pi_l \,(k\neq l).
\end{equation}
\end{proposition}

More generally, a determinantal sampling design has \textit{negative associations} (\citet{lyons2003determinantal}). In particular, for disjoint subsets $A$ and $B$ it holds that
\begin{equation*}\label{EqNegAssoc}
pr(A\cup B \subseteq \mathbb{S})\leq pr(A \subseteq \mathbb{S})pr(B \subseteq \mathbb{S})
\end{equation*}

It was shown recently that determinantal point processes actually enjoy the \textit{strong Rayleigh property} (\citet{borcea2009negative}, \citet{pemantle2014concentration}), a technical property stronger than negative association. This property can be defined in terms of the localization of the zeros of the generating function of the process. These two properties (negative association, strong Rayleigh property) proved very useful for the study of statistics of determinantal processes (\citet{yuan2003central}, \citet{branden2012negative}, \citet{pemantle2014concentration}). Some results will be used in Section \ref{secEst}.

\subsection{Sample size}

Of major importance to statisticians is the sample size of the random sample. It is for instance very common in practice to work with fixed size samples, that is with samples whose size is non-random and given. The sample size of a determinantal random sample follows from Theorem $7$ in \citet{hough2006determinantal}. For a set $A$, let $\sharp A$ denotes its cardinal and for a Hermitian matrix $K$, let $Sp(K)=\{\lambda_i, i\in N\}$ be the set of eigenvalues of $K$ (with their multiplicities).

\begin{theorem}[Sample size]\label{thtaille}
Let $\mathbb{S}\sim DSD(K)$. Then the random variable $\sharp \mathbb{S}$ has the law of a sum of $N$ independent Bernoulli variables $B_1,\cdots,B_N$ of parameters $\lambda_1,\cdots,\lambda_N,$ the elements of $Sp(K)$.
\end{theorem}

\begin{corollary}[Sample size (2)]\label{corEspTaille}
Let $\mathbb{S}\sim DSD(K)$. Then
\begin{enumerate}
\item $E\left(\sharp\mathbb{S}\right)=tr(K)$. 
\item $var(\sharp\mathbb{S})=tr(K-K^2)=\displaystyle\sum_{k \in Sp(K)} \lambda_k(1-\lambda_k)=\displaystyle\sum_{k,l\in U} \Delta_{kl}$.
\item $pr(\mathbb{S}=\emptyset)=0$ iff $1\in Sp(K)$.
\item The total number of points of $DSD(K)$ is less than or equal to $rank(K)$.
\item $DSD(K)$ is a \textit{fixed size determinantal sampling design} iff $K$ is a projection matrix.
\end{enumerate}
\end{corollary}

\begin{proof}
Let $\pi=(\pi_1,\ldots,\pi_N)\,^T$ be the vector of first inclusion probabilities. Then $var(\sharp\mathbb{S})=\sum_{k,l\in U}  \Delta_{kl}$  (see  \citet{sarndal2003model}).
The other results follow directly from Theorem \ref{thtaille} and the spectral decomposition of Hermitian matrices.
\end{proof}

Recall that in case of fixed size sampling designs we have the formula (see \citet{sarndal2003model}): for any $k\in U$, 
\begin{equation}\label{eqSDtaillefixe}
\sum_{l\in U} \Delta_{kl}=0.
\end{equation}

\subsection{Additional properties}

We give here some other general probabilistic results on determinantal sampling designs and their interpretation in terms of sampling theory. We refer to \citet{lyons2003determinantal} and \citet{hough2006determinantal} for their probabilistic versions. 

\begin{proposition}[Complementary sample]\label{propComplement}
Let $\mathbb{S}\sim DSD(K)$. The complementary random sample $\mathbb{S}^c$ is a determinantal random sample with kernel $I_N-K$.
\end{proposition}

\begin{proposition}[Domain]\label{propDomain}
Let $\cP_K$ be a determinantal sampling design on $U$ with kernel $K$, and let $A\subseteq U$ be a subpopulation (or domain). Then the restriction $DSD(K)_{|A}$ of $DSD(K)$ to $A$ is determinantal sampling design on $A$ with kernel $K_{|A}$, the submatrix of $K$ whose rows and columns are indexed by $A$:
\[DSD(K)_{|A}=DSD(K_{|A}).\]
\end{proposition}

\begin{proposition}[Stratification]
Let $\{U_1,\ldots, U_H\}$ be a partition of $U$ into $H$ strata. $DSD(K)$ is stratified iff the matrix $K$ admits a block diagonal decomposition relative to these strata, that is $k\in U_h, l\in U_{h'}, h\neq h'$ implies $K_{kl}=0$.
\end{proposition}

By using the inclusion-exclusion principle, \citet{lyons2003determinantal} shows that the probabilities of disjunction are also given by a determinant (Theorem 5.1 Equation (5.2) for fixed size designs and Equation (8.1) for random size designs).

\subsection{Algorithm}

A general algorithm for simulating a determinantal sampling design is provided in \citet{hough2006determinantal},
including a proof of its validity in a very general setup. Other implementations of this algorithm can be found in \citet{scardicchio2009statistical} and \citet{lavancier2015determinantal}. We consider the latter since it is more suitable and efficient when $N$ is large and $K$ can be written as $V\overline{V}\,^T$, a situation that we will often encounter in Sections \ref{Sec:ConstructingEqual} and \ref{Sec:ConstructingUnequal}.

The first algorithm samples from fixed size determinantal sampling designs.
Let $K=V\overline{V}\,^T$ be a projection matrix, and $v_k^T$ be the $k^{th}$ line of $V$.

\begin{algo}\label{alg1} $\text{ }$\\ 
\begin{itemize}
\item Sample one element $k_n$ of $U$ with probabilities $\Pi_k^n=||v_k||^2/n$, $k \in U$.
\item Set $e_1=v_{k_n}/||v_{k_n}||.$
\item For i = (n-1) to 1 do:
\begin{itemize}
\item sample one $k_i$ of $U$ with probabilities 
$\Pi_k^i=\frac{1}{i}[||v_k||^2-\sum_{j=1}^{j=n-i}|\overline{e_j}^Tv_k|^2]$, $k \in U$,
\item set $w_i=v_{k_i}-\sum_{j=1}^{j=n-i}|\overline{e_j}^Tv_{k_i}e_j|$ and $e_{n-i+1}=w_i/||w_i||$.
\end{itemize}
\item End for.
\item Return $\{ k_1,\cdots,k_n\}.$
\end{itemize}
The resulting sample is a realization of $DSD(K)$.
\end{algo}

The next algorithm describes a procedure to sample from any determinantal sampling design, by expressing it as a mixture of fixed size sampling designs (Theorem 7 in \citet{hough2006determinantal}).

Let $K$ be a contracting matrix with rank one decomposition $K=\sum_{i=1}^{N} \lambda_i v_i\overline{v}_i^T.$
\begin{algo}\label{alg2}$\text{ }$\\
\begin{enumerate}
\item Simulate a vector $b$ whose components are independant Bernoulli variables with parameter $\lambda_1,\cdots,\lambda_N,$.the elements of $Sp(K).$
\item Construct the projection matrix $K_b=\sum_{i=1}^{N} b_i v_i\overline{v}_i^T$.
\item Sample from $DSD(K_b)$ by Algorithm \ref{alg1}. 
\end{enumerate}
The resulting sample is a realization of $DSD(K)$.
\end{algo}

\section{Estimation of a total}\label{secEst}

\subsection{Linear estimators and the Horvitz-Thompson estimator}
Let $y=(y_1,\cdots, y_N)\,^T$ be a variable of interest on the population $U=\{1,\cdots,N\}$. Typical parameters to estimate are the total $t_y=\sum_{k\in U} y_k$, sum of the values of the variable of interest $y$ over the whole population, or the mean value $m_y=t_y/N$. An estimator of $t_y$ is called linear and homogeneous if there exist weights $w_{k}(\mathbb{S}), k\in U$ (that may depend on the sample) such that the estimator writes
\begin{equation*}
\hat{t}_{yw}=\displaystyle\sum_{k\in \mathbb{S}} w_k(\mathbb{S})y_k.\end{equation*} 
When the weights do not depend on the sample, the Mean Square Error ($\mathrm{MSE}$) decomposes as: 

\begin{eqnarray}\label{EqMSE1}
\mathrm{MSE}(\hat{t}_{yw})&=& \overbrace{\underset{k\in U}\sum\underset{l\in U}\sum w_kw_ly_ky_l\Delta_{kl}}^{\mathrm{Variance}} + \left[\overbrace{ \underset{k \in U}\sum(w_k\pi_k-1)y_k}^{\mathrm{Bias}}\right]^2\label{eqVarDelta}\\
&=&\underset{k\in U}\sum w_kw_ly_ky_l(K_{kk}(1-K_{kk}))-\underset{k\in U}\sum\underset{l\neq k}\sum w_kw_ly_ky_l|K_{kl}|^2\nonumber\\
& &+ \left[\underset{k \in U}\sum(w_k\pi_k-1)y_k\right]^2\label{eqVarDelta2}
\end{eqnarray} 

where $\Delta_{kl}$ is defined by Equation \eqref{deltakld}.

Obviously, the only unbiaised estimator should satisfy $w_k=\pi_k^{-1}$, for all $k\in U$. The corresponding estimator, 
\begin{equation*}
\hat{t}_{yHT}=\displaystyle\sum_{k\in \mathbb{S}} \pi_k^{-1}y_k,\end{equation*} is known as the  Horvitz-Thompson estimator (\citet{horvitz1952generalization}). 

\subsection{Mean Square Error}
In the case of a determinantal sampling design, 
the $\mathrm{MSE}$ of an homogeneous linear estimator of the total $t_y$ of a variable of interest $y$ admits algebraic and geometric formulations. They enable us to provide necessary and sufficient conditions for a perfect estimation of the total of auxiliary variables. 

We introduce the following notations. For a vector $x$, $D_x$ denotes the diagonal matrix with diagonal~$x$. For any two matrices $A,B\in\mathcal{M}_{N}(\mathbb{\CC})$, $\langle A,B\rangle= tr(\overline{A}\,^T B)=\sum_{k,l} \overline{a_{k,l}}b_{k,l}$ denotes the canonical scalar product on $\mathcal{M}_{N}(\mathbb{\CC})$. The associated Frobenius norm is denoted by $|A|$.
We also define $z=w*y$ (Schur-Hadamard product) and diagonal matrices $Z=D_{w*y}$, $Z^{1/2}=D_{\sqrt{w*y}}$ where the square root is taken in the complex sense for negative $y$.
Finally, we pose $\langle\langle A,B\rangle\rangle=\langle \overline{Z^{1/2}}\,^TAZ^{1/2}, Z^{1/2}B\overline{Z^{1/2}}\,^T\rangle$. Note that $Z^{1/2}=(Z^{1/2})\,^T$ and $\overline{Z}=Z$, two equalities that we will use thoroughly in the rest of this section. 

\begin{proposition}[Algebraic and Geometric forms of the $\mathrm{MSE}$]\label{PropAlg}
The $\mathrm{MSE}$ of $\hat{t}_{yw}$ satisfies
\begin{eqnarray}
\mathrm{MSE}(\hat{t}_{yw})&=& (\omega*y)\,^T \Delta
(\omega*y) + [e\,^T(K*I_N)(\omega*y)-e\,^Ty]^2 \label{eqMSE}\\
&=&\langle\langle I_N-K,K\rangle\rangle + [\langle D_y, K D_w-I_N \rangle]^2 \label{eqMSEG}
\end{eqnarray} 
and, in the case of the Horvitz-Thompson estimator,
\begin{eqnarray*}
\mathrm{MSE}(\hat{t}_{yHT})=var(\hat{t}_{yHT})&=& (\pi^{-1}*y)\,^T((I_N-K)*\overline{K}
)(\pi^{-1}*y)\label{eqMSEHT}\\
&=&\langle\langle I_N-K,K\rangle\rangle. \label{eqMSEGHT}
\end{eqnarray*} 
\end{proposition}

\begin{proof}
These formulas follow from the classical equality $tr(AB)=tr(BA)$ and the following equality relating the trace on the Schur-Hadamard product \citet{horn1991topics}: for any two vectors $x,y$ and any two matrices $A,B$ it holds that
$$\overline{x}\,^T A*By=tr(\overline{D_x}AD_yB\,^T).$$ 
\end{proof}

The bilinear form $\langle\langle .,.\rangle\rangle$ is indefinite in general. However, it holds by Moutard-Fejer's Theorem (\citet{de2006aspects} Appendix A) that for any two positive semidefinite matrices $A$ and $B$, \begin{equation*}
\langle\langle A,B\rangle\rangle=\langle \overline{Z^{1/2}}\,^TAZ^{1/2}, Z^{1/2}B\overline{Z^{1/2}}\,^T \rangle\geq 0,\end{equation*}
since $ \overline{Z^{1/2}}\,^TAZ^{1/2}$ and $Z^{1/2}B\overline{Z^{1/2}}\,^T$ are positive semidefinite.

Recently, \citet{devillecomment} raised the following question. For a given vector $y$, when can we estimate perfectly (without error, $\mathrm{MSE}=0$) the total $y$, using a sampling design with fixed first order inclusion probabilties (and an homogeneous linear estimator)? Using the previous equations, we provide a necessary and sufficient condition within determinantal sampling designs. Obviously, the estimator must be unbiaised. Therefore we consider the Horvitz-Thompson estimator only ($w_k=\pi_k^{-1}, \, k=1,\ldots,N$), and positive first order inclusion probabilities.

\begin{theorem}[Perfect Estimation]\label{thmQDeville}
Assume $y$ takes only non-zero values, and let $DSD(K)$ be a determinantal sampling design. Let $\alpha_1, \ldots, \alpha_q,$ be the distinct values of $\{\frac{y_k}{\Pi_k},\, k=1,\ldots, N\}$, and $A_j, j=1, \ldots, q$ be the associated sets of indices $k$ such that $\frac{y_k}{\Pi_k}=\alpha_j$. \\
Then the following statements are equivalent:
\begin{enumerate}
\item The total $t_y$ is perfectly estimated ($\mathrm{MSE}=0$) by $\hat{t}^{HT}_y$,
\item $K$ is a projection with positive diagonal that commutes with $Z$,
\item $DSD(K)$ is a stratified determinantal sampling design with strata $A_j,\, j=1,\ldots, q$, of fixed size within each stratum. 
\end{enumerate}
\end{theorem} 

\begin{proof}
$\,$\\
\begin{enumerate}
\item [$1\Rightarrow 2$]
By Moutard-Fejer's Theorem, it holds that for any two semidefinite matrices $A$ and $B$, $tr(AB)\geq 0$ with equality iff $AB=0$.
Assume $\mathrm{MSE}(\hat{t}_{yHT})=0$. Then $tr\left(\overline{Z^{1/2}}\,^T(I_N-K)Z^{1/2} Z^{1/2}K\overline{Z^{1/2}}\,^T\right)=0$. As $\overline{Z^{1/2}}\,^T(I_N-K)Z^{1/2}$ and $Z^{1/2}K\overline{Z^{1/2}}\,^T$ are semidefinite, then $\overline{Z^{1/2}}\,^T(I_N-K)ZK\overline{Z^{1/2}}\,^T=0$. Multiplying on the left and on the right by $\overline{Z^{-1/2}}$ yields $ZK=KZK$ and taking the conjugate transpose gives $ZK=KZK=KZ$. Thus $K$ and $Z$ commute It also follows that $ZK^2=ZK$. By multiplying the equality on the left by $Z^{-1}$ we get $K^2=K$, and $K$ is a projection. 
\item [$2\Rightarrow 3$] Reorder the population by strata. Then the commutant of $Z$ is the set of block diagonal matrices with respect to these stratas, and $K$ is block diagonal. As $K$ is also a projection, each block is actually a projection, and $DSD(K)$ is of fixed size within each stratum. 
\end{enumerate}
$3\Rightarrow 1$ is straightforward.
\end{proof}

Finally, we provide an alternative view on the variance that comes from the general theory of point processes and spatial statistics. The quantity $\sum_{k\neq l} w_kw_ly_ky_l|K_{kl}|^2$ can be interpreted as a ponderated measure of global repulsiveness for point processes on a discrete space (\citet{biscio2016quantifying} and \citet{lavancier2015determinantal} in the continuous setting). As determinantal point processes are repulsive, we then expect DSDs to achieve small variance within all sampling designs. This is validated by our empirical studies in Section \ref{SecSim}. In the next section, we consider the problem of minimization of this variance.

\subsection{Statistical properties of the estimator}

The classical settings for the study of asymptotic properties are either the superpopulation models (\citet{deming1941interpretation}, \citet{cassel1977foundations} chapter 4), or the models of nested (finite) populations as described by Isaki and Fuller \citet{isaki1982survey}. We consider this second setting here. In particular, $(U_N, N\in \NN)$ is a nested sequence of finite populations ($U_N\subseteq U_{N+1}$). The variable of interest $y^N$ may depend on $N$, $(y^N, N\in \NN)$ is a sequence of vectors of size $N$. Also $(w^N, N\in \NN)$ is a sequence of positive vectors of size $N$. In all this section, $(\cP_N, N\in \NN)$ is a sequence of determinantal sampling designs on the populations $U_N$ with kernel $(K^N, N\in \NN)$, whose diagonal terms are positive, and $(\hat{t}^N_{yw}, N\in \NN)$ is the sequence of associated linear estimators of $t_{y^N}$ with weights $w^N$. To simplify notations, we consider as before $U_N=\{1,\ldots, N\}$, and omit the superscript $(.)^N$, writing $y$, $w$, $K$ and $\hat{t}_{yw}$ instead of $y^N, w^N, K^N$ and $\hat{t}^N_{yw}(=\hat{t}^N_{y^Nw^N})$.

We focus successively on consistency, central limit theorems and concentration/deviation inequalities.

In this setting, most results about consistency concern the mean square convergence of the Horvitz-Thompson estimator of the mean $m_y=t_y/N$, see \citet{isaki1982survey}, \citet{robinson1982convergence}, \citet{dol1996matrix} in the case of fixed size sampling designs and \citet{cardot2010properties}, \citet{chauvet2014note} in the general case. A classical condition within these references is that the sequence $\frac{1}{N}\sum_{k=1}^N (\pi_{k})^{-2}y_k^2$ is bounded. Using Schur's Theorem \citet{schur1911bemerkungen} on semidefinite matrices we improve the previous condition for determinantal sampling designs. The theorem also applies to other linear homogeneous estimators than the Horvitz-Thompson one. We pose $\hat{m}_{yw}=\hat{t}_{yw}/N$.

\begin{theorem}[Mean-square convergence]\label{ThCVMQ} 
If
\begin{enumerate}
\item $\displaystyle\sum_{k=1}^N K_{kk} \left(1-\frac{1}{K_{kk}w_k}\right)^2=O(1)$,
\item $\frac{1}{N^2}\displaystyle\sum_{k=1}^N K_{kk}(w_ky_k)^2 \underset{N\to \infty}{\longrightarrow} 0$,
\end{enumerate}
then $(\hat{m}_{yw}-m_y)$ tends to $0$ in mean square.\\
In particular a sufficient condition for the convergence of $(\hat{m}_{y}^{HT}-m_y)$ towards $0$ in mean square is 
$$\frac{1}{N^2}\displaystyle\sum_{k=1}^N \frac{y_k^2}{K_{kk}} \underset{N\to \infty}{\longrightarrow} 0.$$
\end{theorem}

\begin{proof}
By Proposition \ref{PropAlg} 
\begin{equation*}
\mathrm{MSE}(\hat{t}_{yw})= (w*y)\,^T((I_N-K)*\overline{K}
)(w*y) + [e\,^T(I_N*K)(w*y)-e\,^Ty]^2 \end{equation*}
As the matrices $I$, $K$, $I-K$ and $\overline{K}$ are positive semidefinite, it holds that $(I-K)*\overline{K}$, $I*\overline{K}$ and $K*\overline{K}$ are positive semidefinite by Schur Theorem. Since $(I-K)*\overline{K}=I*\overline{K}-K*\overline{K}$ then it also holds that $(I-K)*\overline{K}\leq I*\overline{K}$ for the partial order on positive semidefinite matrices. It follows that 
\begin{eqnarray*}
(w*y)\,^T(I-K)*\overline{K})(w*y) &\leq & (w*y)\,^T (I*\overline{K}) (w*y)\\
&\leq & \sum_{k\in U}(w_ky_k)^2 K_{kk},
\end{eqnarray*} 
Moreover the bias satisfies 
\begin{eqnarray*} [e\,^T(I_N*K)(w*y)-e\,^Ty]^2 &=& \left(\sum_{k\in U} (K_{kk}-\frac{1}{w_k}) (w_ky_k)\right)^2\\
&=& \left(\sum_{k\in U} (\sqrt{K_{kk}}-\frac{1}{\sqrt{K_{kk}}w_k}) (\sqrt{K_{kk}}w_ky_k)\right)^2\\
&\leq & \left(\sum_{k\in U} (\sqrt{K_{kk}}-\frac{1}{\sqrt{K_{kk}}w_k})^2 \right) \left(\sum_{k\in U} (K_{kk}(w_ky_k)^2\right)
\end{eqnarray*} 
by Cauchy-Schwartz-inequality. From these inequalities we get 
\begin{equation*}
E\left((\frac{\hat{t}_{yw}-t_y}{N})^2\right)\leq \left(1+\sum_{k=1}^N K_{kk}(1-\frac{1}{K_{kk} w_k})^2 \right) \frac{1}{N^2}\left(\sum_{k=1}^N K_{kk}(w_ky_k)^2\right)
\end{equation*}
which goes to $0$ by assumptions. This completes the proof.
\end{proof}

Regarding equal probability determinantal sampling designs with expected size $\mu$ ($\pi_k=\mu/N$ for all $k$) and a bounded variable $y$, a sufficient condition for convergence of the Horvitz-Thompson estimator of the mean is simply $\mu\to \infty$. More generally

\begin{corollary}
Set $\mu=trace(K)$. If
\begin{enumerate}
\item there exists $c>0$, such tht for all $N\in \NN$ and all $k\leq N$, $c\frac{\mu}{N}\leq K_{kk}$,
\item the sequence $(\frac{1}{N}\displaystyle\sum_{k=1}^N y_k^2, N\in \NN)$ is bounded,
\item the expected size of the samples $\mu\to \infty$.
\end{enumerate}
Then $(\hat{m}_{y}^{HT}-m_y)\to 0$ in mean square.
\end{corollary}
The second assumption appears for instance in \citet{robinson1982convergence}.\\


Apart consistency, some authors have considered the existence of central limit theorems for sampling designs. However, this proves generally a difficult task even for means or totals, and existing results either focus on a particular class of sampling designs (equal probability sampling designs: \citet{erdos1959central}, \citet{hajek1960limiting}, rejective Poisson sampling: \citet{hajek1964asymptotic}), or assume entropy conditions (\citet{Berger1998rate}). Assuming only that the determinantal sampling design is ``random enough'', we obtain a central limit theorem by applying the results of Soshnikov (\citet{soshnikov2000determinantal},\citet {soshnikov2002gaussian}). These articles contain several theorems on the asymptotic normality of functionals of determinantal point processes. Theorem $1$ on linear statistics of bounded measurable functions in \citet{soshnikov2002gaussian} can be applied straightforwardly to the study of determinantal sampling designs and their associated linear homogeneous estimators (whose weight do not depend on the random sample).

\begin{theorem}[Central Limit Theorem]\label{TCL}
Define for all $N\in \NN$ the homogeneous linear estimators
\begin{equation*}
\hat{t}_{yw}=\displaystyle\sum_{k\in \mathbb{S}} w_k y_k \text{ and } \hat{t}_{|y|w}=\displaystyle\sum_{k\in \mathbb{S}} w_k|y_k|
\end{equation*}
If the variance $var(\hat{t}_{yw})\to +\infty$ as $N\to \infty$ and if 
\begin{equation*}
\sup_{1\leq k\leq N}|w_ky_k|=\textit{o}\left(var(\hat{t}_{yw})\right)^{\epsilon} \text{ and } E(\hat{t}_{|y|w})=\textit{O}\left(var(\hat{t}_{yw})\right)^{\delta} 
\end{equation*}
for any $\epsilon>0$ and some $\delta>0$, then
\begin{equation*}
\frac{\hat{t}_{yw}-E(\hat{t}_{yw})}{\sqrt{var(\hat{t}_{yw})}}\stackrel{law}{\rightarrow}\mathcal{N}(0,1).
\end{equation*}
\end{theorem}

The assumptions of the theorem call for some comments. The assumption $var(\hat{t}_{yw})\to +\infty$ is natural to get a CLT, but a lower bound on the variance is given by the smallest eigenvalue of $(I-K)*K$, that is $0$ for instance for fixed size sampling designs. The two other assumptions are more technical. We present a specific case where they are met.

\begin{corollary}\label{CorTCL} 
If for some $a,b>0$, $\sup_{1\leq k\leq N}|w_k y_k|=\textit{O}\left(\log(N)^a\right)$ and $N^b =\textit{O}\left(var(\hat{t}_{yw})\right)$ then
\begin{equation*}
\frac{\hat{t}_{yw}-E(\hat{t}_{yw})}{\sqrt{var(\hat{t}_{yw})}}\stackrel{law}{\rightarrow}\mathcal{N}(0,1).
\end{equation*}
\end{corollary}

This applies in particular to the Horvitz-Thompson estimation of a bounded variable with $\pi_k>c$.

As usual, we can replace the true variance $var(\hat{t}^N_{yw})$ by any weakly consistent estimator of this variance, using Slutsky theorem. A classical variance estimator is the Horvitz-Thompson estimator of the variance (\citet{horvitz1952generalization}): 
\begin{equation*} 
\widehat{var}_{HT}(\hat{t}_{yw})= \underset{k\in \mathbb{S}}\sum\underset{l\in \mathbb{S}}\sum \frac{w_ky_k w_l y_l}{\pi_{kl}}\Delta_{kl}.
\end{equation*}
In case of fixed size sampling designs, an alternative formula for the variance can be used and we get the Sen-Yates-Grundy estimator (\citet{yates1953selection}, \citet{sen1953estimate}):
\begin{equation*}
\widehat{var}_{SYG}(\hat{t}_{yw})= \frac{1}{2}\underset{k\in \mathbb{S}}\sum\underset{l\in \mathbb{S}}\sum \frac{(w_k y_k-w_l y_l)^2}{\pi_{kl}}\Delta_{kl}
\end{equation*}
which is itself a Horvitz-Thomson estimator (but of a different sum). The first and second order inclusion probabilities of these new sampling designs can be calculated by means of the matrix $K$ associated to $\mathbb{S}$, and we can use the classical criteria of convergence of the Horvitz-Thompson estimator \citet{isaki1982survey}, \citet{robinson1982convergence}, \citet{dol1996matrix}, \citet{cardot2010properties}, \citet{chauvet2014note} that depend only on these inclusion probabilities. \\


As previously recorded, from a very different perspective, the work of \citet{Berger1998rate} proves asymptotic normality for fixed size sampling designs under asymptotically maximal entropy conditions (the sampling is asymptotically rejective). Recently, the asymptotic normality has also been studied for more general classes of processes (that include the determinantal ones): processes with negative or positive associations (\citet{patterson2001limit}, \citet{yuan2003central}), and processes that satisfy the strong Rayleigh property \citet{branden2012negative}.
We adapt here Theorem 2.4 of \citet{patterson2001limit} in the case of the Horvitz-Thompson estimator of the total based on determinantal sampling designs. 
The variance of the Horvitz-Thompson estimator decomposes as \[
var(\hat{t}_{yHT})= \overbrace{\underset{k\in U}\sum y_k^2(K_{kk}^{-1}-1)}^{\text{Poisson\, contribution}}-\overbrace{2\underset{k\in U}\sum\underset{l< k }\sum \frac{y_ky_l}{\pi_k\pi_l}|K_{kl}|^2}^\text{off-diagonal contribution}\] 
Set $s^2=\underset{k\in U}\sum y_k^2(K_{kk}^{-1}-1)$, $r=\underset{k\in U}\sum\underset{l< k }\sum \frac{y_ky_l}{\pi_k\pi_l}|K_{kl}|^2$ and $C=\sup_{1\leq k\leq N}|\pi_k^{-1}y_k|$.
\begin{theorem}
If $s^2\to \infty$, $r=o(s^2)$ and $C=o(s)$, then
\[\frac{\hat{t}_{yHT}-t_y}{s}\stackrel{law}{\rightarrow}\mathcal{N}(0,1).\]
\end{theorem}

For processes satisfying the strong Rayleigh property, \citet{pemantle2014concentration} recently proved concentration and deviation inequalities that extend those of \citet{lyons2003determinantal} for the number of points of determinantal processes in a subdomain. Their application to sampling theory allows derivation of the following finite distance results.
 
\begin{theorem}[Deviation and concentration inequalities]
$\,$\\
Set $\mu=trace(K)$ and set $C=\sup_{1\leq k\leq N}|w_ky_k|.$ For all $a>0$,
\begin{eqnarray*}
pr(\hat{t}_{yw}-E(\hat{t}_{yw}) > a)\leq 3\exp\left(-\frac{a^2}{16\left(aC+2\mu C^2\right)}\right),\\
pr(|\hat{t}_{yw}-E(\hat{t}_{yw})| > a)\leq 5\exp\left(-\frac{a^2}{16^2\left(aC+2\mu C^2\right)}\right).
\end{eqnarray*} 
Moreover, if $DSD(K)$ is of fixed size $\mu=n$, then
\begin{eqnarray*}
pr(\hat{t}_{yw}-E(\hat{t}_{yw})> a)\leq \exp\left(-\frac{a^2}{8nC^2}\right),\\
pr(|\hat{t}_{yw}-E(\hat{t}_{yw})| > a)\leq 2\exp\left(-\frac{a^2}{8nC^2}\right).
\end{eqnarray*} 
\end{theorem}

\begin{proof}
Function $s\mapsto \sum_{k=1}^N w_ky_k 1_{\{k \in s\}}$ is $C$-Lipschitz for the Hamming distance. Theorems $3.1$ and $3.2$ of \citet{pemantle2014concentration} apply and yield the stated results. 
\end{proof}

From this concentration inequality, we derive a new criterion for the convergence in probability of $\hat{t}_{yHT}$:

\begin{corollary}\label{CorCardotNew}
If $\frac{\sqrt{trace(K)}}{N}\sup_{1\leq k \leq N}|\frac{y_k}{K(k,k)}|\underset{N\to \infty}{\longrightarrow} 0$ then
$\frac{\left(\hat{t}_{yHT}-t_{y}\right)}{N}\overset{pr}{\underset{N\to \infty}{\longrightarrow}} 0.$
\end{corollary}

\begin{proof}
Set $C_N=\sup_{1\leq k\leq N}\frac{|y_k|}{K_N(k,k)}$, $\mu_N=trace(K_N)$. It holds that
\begin{equation*}
pr(|\hat{t}_{yHT}-t_{y}| > Na)\leq 5\exp\left(-\frac{N^2a^2}{16^2\left(NaC_N+2\mu_N C_N^2\right)}\right)
\end{equation*}
By assumption $C_N=o(N)$ and $\mu_N C_N^2=o(N)^2$, and the right hand term above tends to $0$.
\end{proof}

\section{Constructing equal probability determinantal sampling designs}\label{Sec:ConstructingEqual}

In practice, one mainly uses sampling designs with fixed first order inclusion probabilities. It is thus of crucial importance to be able to build such sampling designs with additional interesting properties, such as a fixed size. The parametric description of determinantal sampling designs proves a formidable tool to achieve this goal. 

Sampling designs with constant first order inclusion probabilities are a particular instance of sampling designs with fixed first order inclusion probabilities, called \textit{equal probability sampling designs}. We consider such designs in this section (next section will be devoted to sampling designs with any prescribed probabilities). We first consider the existence of determinantal sampling designs having the same first and second order probabilities as a SRS. Second, we provide an explicit construction of fixed size and equal probability determinantal sampling designs relying on the $N$-th primary unit roots. We finally relax the fixed size constraint to build new examples of equal probability determinantal sampling designs.

\subsection{$(N,n)$-simple determinantal sampling designs}

SRS is not determinantal in general. This negative result does not however settle the question of the existence of a determinantal sampling design with the same first and second order inclusion probabilities as the SRS of size $n$, that is such that $\pi_k=\frac{n}{N}$ and $\pi_{kl}=\frac{n(n-1)}{N(N-1)} (k\neq l).$

\begin{definition}[$(N,n)$-simple designs]
A sampling design $\cP$ (on $U$ of size $N$) is $(N,n)$-simple if its inclusion probabilities satisfy $$\pi_k=\frac{n}{N}\text{ and }\pi_{kl}=\frac{n(n-1)}{N(N-1)}.$$
\end{definition}

Since the variance of $\sharp \mathbb{S}$ only depends on the first and second order probabilities, which are those of SRS, and as SRS is of fixed size, it follows that a $(N,n)$-simple sampling design is of fixed size $n$. In particular, the kernel of a $(N,n)$-simple determinantal sampling design is a projection of rank $n$.

Applying Equations \eqref{pikd} and \eqref{pikld} to a $(N,n)$-simple determinantal sampling design, we get that its kernel $K$ also satisfies
\begin{equation*}
\left\{\begin{array}{lll} K_{kk}&=& \frac{n}{N}, \\ 
|K_{kl}|^2&=& \frac{n(N-n)}{N^2(N-1)}(k\neq l).\end{array}\right. 
\end{equation*}

Let $F$ be a $(n\times N)$ matrix such that $V=(\frac{n}{N})^{1/2}\overline{F}\,^T$ is an orthonormal basis of the range of $K$ ($K=V\overline{V}\,^T=\frac{n}{N}\overline{F}\,^TF$) ,where $K$ is the kernel of a $(N,n)$-simple determinantal sampling design. It holds that: 
\begin{enumerate}
\item For all $l\in 1,\ldots, N$, $\sum_{k=1}^{n} F_{kl}^2=1$,
\item There exists a nonnegative $\alpha$ such that $|\sum_{j=1}^{n} \overline{F_{jk}}F_{jl}|=\alpha=\sqrt{\frac{N-n}{n(N-1)}}$ ($k\neq l$),
\item $F\overline{F}\,^T=\frac{N}{n}I_n$.
\end{enumerate} 
These properties are exactly those defining an \textit{Equiangular Tight Frame (ETF)} according to \citet{tropp2005complex} and \citet{sustik2007existence}. Thus

\begin{theorem}[$(N,n)$-simple designs and ETFs]\label{ThETF}
$DSD(K)$ is a $(N,n)$-simple sampling design iff $K=\frac{n}{N}\overline{F}\,^TF$, where $F=(f_1,\cdots, f_N)$ is an $ETF$ of $\CC^n$.
\end{theorem}

As a consequence of Theorem \ref{ThETF}, a necessary and sufficient condition for the existence of ETFs would solve the problem of the existence of $(N,n)$-simple determinantal sampling designs. However, such a condition is not known for the moment. But there exist necessary conditions. In Theorem \ref{thETFsimple}, we apply some existing results on ETFs to $(N,n)$-simple determinantal sampling designs (this a only a small part of a rich literature on the subject, going from strongly regular graphs \citet{waldron2009construction} to Gauss sums and finite field theory \citet{strohmer2008note}).

\begin{theorem}\label{thETFsimple}
Let $1<n<N-1$ be two integers.
\begin{enumerate}
\item There exists a $(N,n)-simple$ determinantal sampling design only if $N\leq min\{n^2,(N-n)^2\}$ (\citet{tropp2005complex}).
\item There exists a $(N,n)-simple$ determinantal sampling design with a real kernel $K$ only if $N\leq \min\{\frac{n(n+1)}{2}, \frac{(N-n)(N-n+1)}{2}\}$ (\citet{sustik2007existence} Theorem C).
\item When $N \neq 2n$, a necessary condition of the existence of a $(N,n)$-simple determinantal sampling design with real kernel $K$ is that the following two quantities be odd integers:
$$\alpha=\sqrt{\frac{n(N-1)}{N-n}}, \,\beta=\sqrt{\frac{(N-n)(N-1)}{n}}.$$
When $N=2n$, it is necessary that $n$ be odd and that $N-1$ be the sum of two squares (\citet{sustik2007existence} Theorem A and \citet{casazza2008real} Theorem 4.1)
\end{enumerate}
\end{theorem}

Numerical studies compensate for the absence of general existence conditions. We deduce from \citet{sustik2007existence} and \citet{casazza2008real} the existence of $(N,n)$-simple determinantal sampling designs with real kernel for respectively $N\leq 100$ (\citet{sustik2007existence} Table I) and $n\leq 50$ (\citet{casazza2008real} Table III). Tables II and III in \citet{sustik2007existence} also give the existence of $(N,n)$-simple determinantal sampling designs with complex kernels. Table \ref{simple} summarizes this information for $n<9.$ In the table, the symbol $\mathbb{C}$ indicates that no $(N,n)$-simple determinantal sampling design with real kernel exists, but that one with complex kernel does exist (this justifies \textit{a posteriori} the use of complex matrices). 

\begin{table}[h!]
\begin{center}
\caption{Existence of $(N,n)$-simple determinantal sampling designs, depending on the kernel type (real or complex) for $n<9.$ \label{simple}}
\begin{tabular}{|c|c|c|c|c|c|c|c|c|c|c|c|c|c|c|c|} \hline
$n$&	3&	3&	4&	4&	5&	5&	6&	6&	6&	7&	7&	7&	8&	8&	8\\ \hline
$N$&	6&	7&	7&	13&	10&	11&	11&	16&	31&	14&	15&	28&	15&	29&	57\\ \hline
&	$\mathbb{R}$&	$\mathbb{C}$&	$\mathbb{C}$&	$\mathbb{C}$&	$\mathbb{R}$&	$\mathbb{C}$&	$\mathbb{C}$&	$\mathbb{R}$&	$\mathbb{C}$&	$\mathbb{R}$&	$\mathbb{C}$&	$\mathbb{R}$&	$\mathbb{C}$&	$\mathbb{C}$&	$\mathbb{C}$\\ \hline
\end{tabular}
\end{center}
\end{table}

\begin{example}[(6,3)-simple determinantal sampling design]
Let $$K=\frac{1}{2}\left(\begin{array}{cccccc}
1 & \frac{1}{\sqrt{5}} & \frac{1}{\sqrt{5}}& \frac{1}{\sqrt{5}}& \frac{1}{\sqrt{5}}& \frac{1}{\sqrt{5}}\\ \frac{1}{\sqrt{5}} & 1 & -\frac{1}{\sqrt{5}}& -\frac{1}{\sqrt{5}}& \frac{1}{\sqrt{5}}& \frac{1}{\sqrt{5}}\\ \frac{1}{\sqrt{5}} & -\frac{1}{\sqrt{5}} & 1& \frac{1}{\sqrt{5}}& -\frac{1}{\sqrt{5}}& \frac{1}{\sqrt{5}}\\ \frac{1}{\sqrt{5}} & -\frac{1}{\sqrt{5}} & \frac{1}{\sqrt{5}}& 1& \frac{1}{\sqrt{5}}& -\frac{1}{\sqrt{5}}\\ \frac{1}{\sqrt{5}} & \frac{1}{\sqrt{5}} & -\frac{1}{\sqrt{5}}& \frac{1}{\sqrt{5}}& 1& -\frac{1}{\sqrt{5}}\\ \frac{1}{\sqrt{5}} & \frac{1}{\sqrt{5}} & \frac{1}{\sqrt{5}}& -\frac{1}{\sqrt{5}}& -\frac{1}{\sqrt{5}}& 1 \end{array}\right).$$ $K$ is a projection, and $DSD(K)$ is $(6,3)$-simple. It is not a simple sampling as the samples $\{1,2,3\}$ and $\{4,5,6\}$ do not have the same probabilities ($\frac{1}{8}(1-\frac{3}{5}-\frac{2}{5\sqrt{5}})$ and $\frac{1}{8}(1-\frac{3}{5}+\frac{2}{5\sqrt{5}})$ respectively).
\end{example}

\subsection{Fixed size equal probability sampling designs}\label{secFSEWSD} 

In this section, we construct an explicit family of fixed size, equal probability sampling designs. The matrices involved are special Toeplitz matrices constructed upon primitive $N$th roots of the unity. 

\begin{theorem}\label{propmary}
Let $n,r,N$ be three integers such that $n\leq N$ and $r<N$ with $r,N$ two relatively prime integers. Let $DSD(K^{r,N,n})$ be the determinantal sampling design with kernel $K^{r,N,n}$:
$$\left\{\begin{array}{l}
K^{r,N,n}_{kl}=\frac{1}{N}\frac{\sin(\frac{nr(k-l)\pi}{N})}{\sin(\frac{r(k-l)\pi}{N})}e^{\frac{i r(n-1)(k-l)\pi}{N}},\\
K^{r,N,n}_{kk}=\frac{n}{N}. 
\end{array}\right.$$ $DSD(K^{r,N,n})$ is of fixed size $n$, and its first and second order inclusion probabilities satisfy 
\begin{equation*}
\left\{\begin{array}{lll}
\pi^{r,N,n}_{k} & = & \frac{n}{N}, \\
\pi^{r,N,n}_{kl} &= &\frac{n^2}{N^2} - \frac{1}{N^2}\frac{\sin^2\left(\frac{nr(k-l)\pi }{N}\right)}{\sin^2\left(\frac{r(k-l)\pi }{N}\right)}\,(k\neq l).
\end{array} \right. 
\end{equation*} 
\end{theorem}

\begin{proof}
Let $z=e^{\frac{2i\pi r }{N}}$ be a any primitive $N$th root of the unity with $r,N$ two relatively prime integers. Set $c=n/N$ and define for all $p=0,\ldots,n-1$ the vectors 
$v_p=\frac{\sqrt{c}}{\sqrt{n}}\left(\left(z^p\right)^1,\ldots, \left(z^p\right)^{N}\right)\,^T$.
They define, by construction, an orthonormal family and $K^{r,N,n}=\sum_{p=0}^{n-1} v_p\overline{v_p}\,^T=V\overline{V}\,^T$ is a projection of rank $n$, where $V=(v_0,\cdots, v_{n-1})$. Its diagonal elements satisfy \[P_{kk}^{r,N,n}=\sum_{p=0}^{n-1} V_p(k)\overline{V_p}\,^T(k)=n^{-1}c\sum_{p=0}^{n-1} 1=c\] for all $k=1,\ldots, N$. Its off-diagonal elements satisfy
\begin{eqnarray*}
K^{r,N,n}_{kl}&=\displaystyle\frac{1}{N}\sum_{p=0}^{n-1} z^{(k-l)p}&=\displaystyle\frac{1}{N}\frac{1- z^{(k-l)n}}{1-z^{(k-l)}}\\
&= \displaystyle\frac{1}{N}\frac{1- e^{\frac{2i\pi r(k-l)n}{N}}}{1-e^{\frac{2i\pi r(k-l)}{N}}}&=\displaystyle\frac{1}{N}\frac{\sin\left(\frac{nr(k-l)\pi }{N}\right)}{\sin\left(\frac{r(k-l)\pi}{N}\right)}e^\frac{ir(n-1)(k-l)\pi }{N} \, (k\neq l).
\end{eqnarray*}
The second order inclusion probabilities follows from Equation \eqref{pikld}. 
\end{proof}

Figure \ref{graphe1} shows those probabilities in the following case : $N=18,n=6,r \in (1,5,7).$

\begin{figure}[h]
\begin{center}
\caption{Study of $\cP^{r,18,6}$ for $r \in (1,5,7)$ : $\pi^{r,18,6}_{k}=\pi_k=\frac{6}{18}$}\label{graphe1}
Circle with radius proportional to $\frac{\pi^{r,18,6}_{kl}}{\pi_{k}\pi_{l}}(k\neq l) $ ($0\leq black\leq 0.5$, $0.5\leq blue\leq 0.9$,$0.9\leq red\leq 1$.)
\begin{subfigure}{0.29\textwidth}
\includegraphics[scale=0.29]{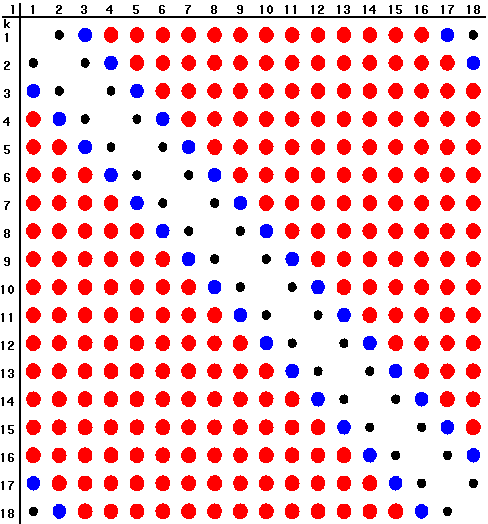}
\caption{$r=1$} \label{fig:r1}
\end{subfigure}
\begin{subfigure}{0.29\textwidth}
\includegraphics[scale=0.29]{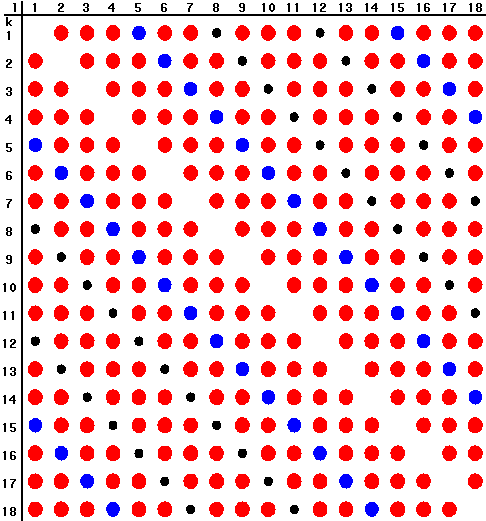}
\caption{$r=5$} \label{fig:r5}
\end{subfigure}
\begin{subfigure}{0.29\textwidth}
\includegraphics[scale=0.29]{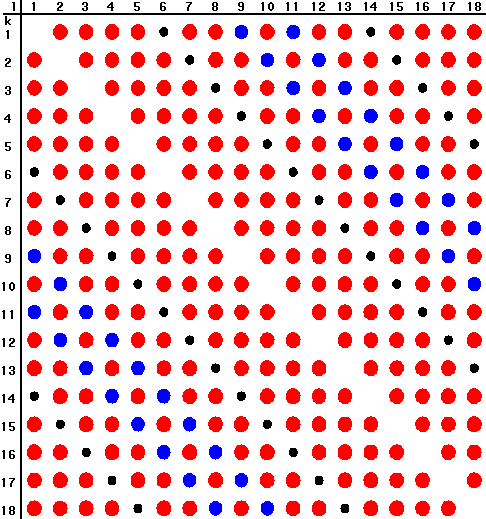}
\caption{$r=7$} \label{fig:r7}
\end{subfigure}
\end{center}
\end{figure}

\begin{corollary}
Let $n, r, N$ be three integers such that $n< N$ and $r<N$ with $r,N$ two relatively prime integers. Then
\begin{equation*}
\displaystyle\sum_{k=1}^{N-1}\frac{\sin^2\left(\frac{nrk\pi}{N}\right)}{\sin^2\left(\frac{rk\pi}{N}\right)}=n(N-n)
\end{equation*}
\end{corollary}

\begin{proof}
We apply Equation \eqref{eqSDtaillefixe} with $l=N$. 
\end{proof}

\subsection{Other equal probability determinantal sampling designs}

Relaxing the fixed size constraint leads to new families of determinantal sampling designs, as shown in the following results. 

\begin{theorem}
Let $n, N$ be two integers such that $0<n<N$, and let $DSD(K^{N,n})$ be the determinantal sampling design with kernel $K^{N,n}$:
$$\left\{\begin{array}{l}
K^{N,n}_{kl}=\frac{N-n}{N(N-1)}\\
K^{N,n}_{kk}=\frac{n}{N}. 
\end{array}\right.$$ $DSD(K^{N,n})$ is of random size $n\geq 1$, and its first and second order inclusion probabilities satisfy 
\begin{equation*}
\left\{\begin{array}{lll}
\pi^{N,n}_{k} & = & \frac{n}{N}, \\
\pi^{N,n}_{kl} &= &\frac{n^2}{N^2} - \frac{1}{N^2}\,(k\neq l).
\end{array} \right. 
\end{equation*} 
\end{theorem}

\begin{proof}
The characteristic polynomial of $K^{N,n}$ can be computed as a Hurwitz determinant: $p(K^{N,n})=(1-\lambda)(\frac{n-1}{N-1}-\lambda)^{N-1}$. $K^{N,n}$ is thus a contracting matrix with $1$ as maximal eigenvalue, and by Corollary \ref{corEspTaille}, $pr(\mathbb{S}=\emptyset)=0$. As $K^{N,n}$ is not a projection, $DSD(K^{N,n})$ is not of fixed size. 
\end{proof}

Actually, $K^{N,n}$ is the mean of the previous kernels $K^{r,N,n}$.

\begin{lemma}
Let $N$ be prime and $n<N$. Then $K^{N,n}=\frac{1}{N-1}\sum_{r=1}^{r=N-1}K^{r,N,n}$.
\end{lemma}

\begin{proof}
$\frac{1}{N-1}\sum_{r=1}^{r=N-1}K^{r,N,n}_{kl}=\frac{1}{N-1}\sum_{r=1}^{r=N-1}\frac{1}{N}\sum_{p=0}^{n-1} z^{(k-l)p}$ $=\frac{1}{N}\frac{1}{N-1}\sum_{p=0}^{n-1}\sum_{r=1}^{r=N-1} z^{(k-l)p}=\frac{N-n}{N(N-1)}$
\end{proof}

This results is also true when $N$ is not prime, but in that case, for some values of $r$, $K^{r,N,n}$ might not be contracting. \\

We conclude this section by providing two general methods to construct (non-fixed size) equal probability determinantal sampling designs. The first one relies on positive definite kernel functions and takes into account auxiliary information in $\mathbb{R}^Q.$ We illustrate the method with the Laplacian kernel, but other kernels are also available (linear, Gaussian).  

\begin{example}[Laplacian kernel]\label{exLaplace}
Set $0<\alpha<1$ and let $x\in \RR^Q$ be an auxiliary variable. For $\beta$ large enough, there exists a determinantal sampling design with first and second order inclusion probabilities 
\begin{equation*}
\left\{\begin{array}{lll}
\pi_{k} & = &\alpha, \\
\pi_{kl} &= & \alpha^2(1-\exp{^{-2\beta(|x_k-x_l|_1)}}) \,(k\neq l).
\end{array}\right. 
\end{equation*} 
Indeed, the Laplacian kernel function $f^{\alpha,\beta}(x)=\alpha\exp^{-\beta|x|_1}$ is positive semidefinite on $R^Q$ with $\alpha>0$ and $\beta>0$. The matrix $L^{\alpha,\beta}$ defined by $L^{\alpha,\beta}_{kl}=f^{\alpha,\beta}(x_k-x_l)$ is thus positive semidefinite. For $\beta$ large enough, its eigenvalues are less than $1$. The quantity $\alpha N$ is the average size of the random sample. If $x_k=x_l (k\neq l)$, then $k$ and $l$ will never be sampled simultaneously. 
\end{example}

The second method relies on (infinite) Hermitian Toeplitz matrices, see for instance \citet{grenander1958szego}. From this theory we directly obtain the following result.
 
\begin{proposition}[Toeplitz Designs]
Let $0\leq f\leq 1$ be a real square integrable function on $[0,2\pi]$. Then for any $N\in \NN$, the matrix $T_N(f)$ with coefficients $T_N(f)(k,l)=\int_{0}^{2\pi} f(\lambda) e^{-i(k-l)\lambda}d\lambda$ is a contracting matrix, with constant diagonal $|f|_1=\int_{0}^{2\pi} f(\lambda) d\lambda$, and $DSD\left(T_N\left(f\right)\right)$ is an equal probability determinantal sampling design of average size $|f|_1\times N$.
\end{proposition}

\section{Constructing fixed size, unequal probabilities determinantal sampling designs}\label{Sec:ConstructingUnequal}

To derive estimators with a low MSE, it is common to work with \textit{unequal probabilities sampling designs}, that is sampling designs  with unequal first order inclusion probabilities. These designs are preferably of fixed size, for both practical and theoretical reasons. Indeed, as shown for instance by Theorem \ref{thmQDeville} in the determinantal case, the total $t_y$ will be perfectly estimated by the Horvitz-Thompson estimator $\hat{t}_{yHT}$ iff the vector $\pi$ of first order inclusion probabilities is proportional to the vector $y$ and the sampling design is of fixed size. In this section we first give a general existence result of fixed size determinantal sampling designs with prescribed inclusion probabilities. Then we discuss the effective construction of such designs. A major improvement compared to existing methods (systematic designs, cube method,...) is that the second order inclusion probabilities are completely known. This allows a precise description of these sampling designs (Corollary \ref{corovin}). 

\subsection{Theoretical approach}
Let $\Pi$ be a vector of size $N$ such that $ 0\leq \Pi_k\leq 1$. There exists a very simple determinantal sampling design satisfying $\pi_k=\Pi_k$ for all $k$: the Poisson sampling \eqref{poisson} with kernel $K^{\Pi}$ defined by $K^{\Pi}_{kk}=\Pi_k$ and $K^{\Pi}_{kl}=0,\, k\neq l$. Unfortunately this design is not of fixed size. According to Corollary \ref{corEspTaille}, constructing a fixed size determinantal sampling design with prescribed inclusion probabilities is equivalent to constructing a projection matrix with a prescribed diagonal. The latter problem is a particular case of the more general issue of constructing Hermitian matrices with prescribed diagonal and spectrum
that has re-attracted attention over the last years (\citet{schur1911bemerkungen}, \citet{horn1954doubly}, \citet{kadison2002pythagorean}, \citet{fickus2013constructing},  \citet{dhillon2005generalized}). 



Obviously, since $\sum_{k=1}^N \pi_k$ is the expected number of points in the sample, a necessary condition to obtain fixed size sampling designs with $\pi=\Pi$ is that $\sum_{k=1}^N \Pi_k$ is an integer. 
The next theorem proves that this is actually a sufficient condition. 
It is a direct consequence of the Schur-Horn Theorem (\citet{horn1954doubly}).

\begin{theorem}[Fixed size DSD with prescribed unequal probabilities (Existence)]\label{thHorn1}
$\,$\\
Let $\Pi$ be a vector of size $N$ such that $0\leq \Pi_k\leq 1$ and $\sum_{k=1}^N \Pi_k=n\in \mathbb{N}$. There exists a determinantal sampling design $DSD(K^{\Pi})$ of fixed size $n$ whose first order inclusion probabilities satisfy $\pi_k=\Pi_k=K^{\Pi}_{kk}$.
\end{theorem}

The original proof of the Schur-Horn theorem is non-constructive. In the next section, we give an explicit construction of the sampling design. 

\subsection{Effective construction}\label{SubSecEC}
Up to now, the effective constructions found in the literature are algorithmic and do not provide a closed form for the matrix $K^{\Pi}$. For instance \citet{kadison2002pythagorean} and \citet{dhillon2005generalized} provide algorithms based on plane rotations whereas \citet{fickus2013constructing} recently describe a Top Kill algorithm using frame theory. 
In the two dimensional case, \citet{dhillon2005generalized} explicitly construct a (real) plane rotation $Q_2$ so that the diagonal vector of $K_2=Q_2AQ_2^T(=Q_2A\overline{Q_2}\,^T)$ equals a prescribed vector $(\Pi_1,\Pi_2)\,^T$, while having the same spectrum as $A$. Assuming (without loss of generality) that $a_{1,1} \leq \Pi_1 \leq \Pi_2 \leq a_{2,2}$, then
$$Q_2 \left( \begin{array}{cc}
a_1 & a^*_{21} \\
a_{21} & a^*_{2} \\
  \end{array}\right ) Q_2^T = \left( \begin{array}{cc}
  \Pi_1 & * \\
* & \Pi_2 \\
  \end{array}\right ),$$ where $$ Q_2 = \left( \begin{array}{cc}
\sin \theta & \cos \theta  \\
-\cos \theta & \sin \theta \\
  \end{array}\right ), $$ 
\begin{eqnarray}\label{matrot}
t &=& \frac{Re a_{21}\pm \sqrt{(Re a_{21})^2-(a_1-\Pi_{1})(a_2-\Pi_{1})}}{a_2-\Pi_{1}}, \nonumber \\
\sin \theta &=& \frac{1}{\sqrt{1+t^2}}, \\
\cos \theta &=& t\sin \theta.  \nonumber
\end{eqnarray}  

Using these plane rotations and the ideas behind the proof of Theorem 7 in \citet{kadison2002pythagorean} we are able to derive a closed-form expression of a specific projection matrix $K^{\Pi}$, denoted $P^{\Pi}$ afterwards, and of its joint inclusion probabilities $\pi_{kl}$ for all $(k,l).$

Let $\Pi$ be a vector of size $N$ such that $0< \Pi_k < 1$ and $\sum_{k=1}^{N}\Pi_k=n \in \mathbb{N}^*$. Set $k_0=0$ and for all integer $r$ such that $1\leq r \leq n$, let  
\begin{itemize} 
\item $1 < k_r \leq N$ be the integer such that $\underset{k=1}{\overset{k_r-1}\sum}\Pi_k< r\text{ and }\underset{k=1}{\overset{k_r}\sum}\Pi_k\geq r\}
$,
\item $\alpha_{k_r}=r-\underset{k=1}{\overset{k_r-1}\sum}\Pi_k$ and $\alpha_k = \Pi_k$ if $k \neq k_r$, 
\item $T_k = \underset{i=k}{\overset{k_{r+1}}\sum}\alpha_i$ for $k_r<k\leq k_{r+1}$,
\item $\gamma_r^{r'}=\sqrt{\underset{j=r+1}{\overset{r'}\prod} \frac{(\Pi_{k_{j}}-\alpha_{k_{j}})\alpha_{k_{j}}}{(1-\alpha_{k_{j}})(1-(\Pi_{k_j}-\alpha_{k_j}))}}$ for $r<r'$, $\gamma_r^{r'}=1$ otherwise. 
\end{itemize}

\begin{figure}[h!]
\begin{center}
\caption{Representation of the various quantities, $n=3$}\label{ExplicationTheorem}
\includegraphics[scale=1]{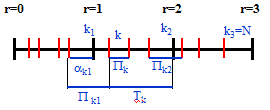}
\end{center}
\end{figure}

Then define a real symmetric kernel $P^{\Pi}$ as follows: 
\begin{itemize}
\item for all $1\leq k\leq N$, $P^{\Pi}_{kk}=\Pi_k$,
\item for all $k>l$ :

\begin{table}[h!]
\begin{center}
\caption{\text{Values of} $P^{\Pi}_{kl}$ : $k>l$}\label{tabpikl}
\begin{tabular}{|c|c|c|} \hline 
& \multicolumn{2}{|c|}{\text{Values of} $l$} \\ \cline{1-3}

  Values of $k$    & $l=k_r$                    & $k_r<l<k_{r+1}$ \\ \hline
$k_{r'}<k<k_{r'+1}$ & $-\sqrt{\Pi_k}\sqrt{\frac{(1-\Pi_l)(\Pi_l-\alpha_l)}{1-(\Pi_l-\alpha_l)}}\gamma_r^{r'}$                      &  $\sqrt{\Pi_k\Pi_l}\gamma_r^{r'}$                                       \\
$k=k_{r'+1}$     &                        $-\sqrt{\frac{(1-\Pi_{k})\alpha_k}{1-\alpha_{k}}}\sqrt{\frac{(1-\Pi_l)(\Pi_l-\alpha_l)}{1-(\Pi_l-\alpha_l)}}\gamma_r^{r'}$  &  $\sqrt{\frac{(1-\Pi_{k})\alpha_k}{1-\alpha_{k}}}\sqrt{\Pi_l}\gamma_r^{r'}$ \\ \hline
\end{tabular}
\end{center}\end{table}

\end{itemize}

\begin{theorem}[Fixed size DSD with prescribed unequal probabilities (Construction)]\label{theovin} 
$\,$\\
The matrix $P^{\Pi}$ is a real projection matrix, and $DSD(P^{\Pi})$ is a fixed size sampling designs with first order inclusion probabilities $\pi_k=\Pi_k, 1\leq k\leq N$.
\end{theorem}

\begin{proof} 
Let $P_0$ be a $(N\times N)$ matrix whose entries are $0$ apart from $(k_r+1,k_r+1)$ entries, $r=0,\cdots,n-1$ whose values are $1$. The proof of Theorem $7$ in \citet{kadison2002pythagorean} consists in constructing the sequence $P^j=W(j)^TP^{j-1}W(j)$, $j=1,\cdots,N$, where $W(j)$ is the unitary operator whose matrix relative to the canonical basis has $\sin \theta_j$ at the ($j,j$) and ($j+1,j +1$) entries, $\cos \theta_j$ and $-\cos \theta_j$ at the $(j,j+1)$ and $(j+1,j)$ entries, respectively, 1 at all other diagonal entries, and $0$ at all other off-diagonal entries. $\theta_j$ is the angle enabling to have $\Pi_j$ at the $(j,j)$ entry of $P^j$, without changing the first $j-1$ and the last $N-j-1$ diagonal entries. 
We build on this proof and on formulas \ref{matrot} to give explicit calculations of $\sin \theta_j$ and we finally end up with the following coefficients for $P^{N}=P^{\Pi}$: 
 \begin{eqnarray}
  P^{\Pi}_{kl}=&\qquad\sin \theta_k (\alpha_l-1)\sin \theta _l \underset{j=l}{\overset{k=k-1}\prod} \cos \theta_j & ( l=k_r < k) \\   P^{\Pi}_{kl}=&\qquad\sin \theta_k T_l \sin \theta_l \underset{j=l}{\overset{k=k-1}\prod} \cos \theta_j & ( l \neq k_r < k)  
 \end{eqnarray}
where $\sin \theta_j = \sqrt{\Pi_jT_j^{-1}}$ and $\cos \theta_j = \sqrt{T_{j+1}T_j^{-1}}$ when $j=k_r$. $\sin \theta_j=\sqrt{(1-\Pi_j)(1-\alpha_j)^{-1}}$ and $\cos \theta_j=\sqrt{(\Pi_j-\alpha_j)(1-\alpha_j)^{-1}}$ otherwise. 
 \end{proof}

The exact knowledge of the coefficients $P^{\Pi}_{kl}$ enables a precise characterization of the sampling designs so constructed.

\begin{corollary} \label{corovin}
Let $P^{\Pi}$ be the matrix previously constructed, and $DSD(P^{\Pi})$ the associated sampling design. 
\begin{enumerate}
\item If $(k,l) \in ] k_r+1,k_{r+1}-1 [^2$ then $\pi_{kl}=0$.
\item If $i \in ] k_r+1,k_{r+1}-1 [$, $j=k_{r+1}$, $k \in ] k_{r+1}+1,k_{r+2}-1 [$ then $\pi_{ijk}=0$ .
\item Set $B_r=[1,k_{r}+1]$. Then $1$ is an eigenvalue of multiplicity $r$ and $0$ an eigenvalue of multiplicity $k_r-r$ of $K_{|B}$ : the random sample $\mathbb{S}$ has $r$ or $r+1$ elements in $B_r$ ($r\leq \sharp (\mathbb{S}\cap B_r)\leq r+1$). 
\item If $k-l$ is large then $P^{\Pi}_{kl}\approx 0$, and the events $\{k\in \mathbb{S}\}$ and $\{l\in \mathbb{S}\}$ are asymptotically independent. In practice $\pi_{kl}\approx \Pi_k\Pi_l$ also holds for small values of $k-l$.
\item Let $r_{1},\cdots ,r_H$ be the set of values of $1\leq r \leq n$ such that $\sum^{k_r}_{k=1}\Pi_k= r$, and set $r_0=0$. Then $DSD(P^{\Pi})$ is stratified with $H$ strata $]k_{r_{h-1}},k_{r_{h}}]$. In particular, for constant $\Pi_{k}$ and if $n$ divides $N$ then the sampling design picks exactly one element according to the $\Pi_k$ in each of the $n$ strata. 
\end{enumerate}
\end{corollary}
 
The first two points follow from the calculations of the respective $2\times 2$ and $3\times 3$ determinant. Using Sarrus rule, we get that $\det(P^{\Pi}_{|ijk})=0$, where:
$$P^{\Pi}_{|ijk}=\left (
\begin{array}{ccc}

  \Pi_i & \sqrt{\frac{(1-\Pi_{j})\alpha_j}{1-\alpha_{j}}\Pi_i} & \sqrt{\Pi_k\Pi_i\frac{(\Pi_j-\alpha_j)\alpha_j}{(1-\alpha_j)(1-(\Pi_j-\alpha_j))}}\\ 

 \sqrt{\frac{(1-\Pi_{j})\alpha_j}{1-\alpha_{j}}\Pi_i} & \Pi_j & -\sqrt{\Pi_k\frac{(1-\Pi_j)(\Pi_j-\alpha_j)}{(1-(\Pi_j-\alpha_j))}} \\ 

 \sqrt{\Pi_k\Pi_i\frac{(\Pi_j-\alpha_j)\alpha_j}{(1-\alpha_j)(1-(\Pi_j-\alpha_j))}} & -\sqrt{\Pi_k\frac{(1-\Pi_j)(\Pi_j-\alpha_j)}{(1-(\Pi_j-\alpha_j))}} & \Pi_k  \\ 
\end{array} \right ).
$$ 
$P^{\Pi}_{|B_r}$ has the same spectrum as $P^0_{|B_r}$ proving point $3$. 
Point $4$ follows from the expression of $P^{\Pi}_{kl}$ as a product of cosine. Finally, point $5$ follows from the following fact: if ${\sum^{k_r}_1}\Pi_k= r$ then $\Pi_{k_r}-\alpha_{k_r}=0$ and $P^{\Pi}_{kl}=0$ for $k,l$ in different strata.\\

These results call for some comments. The construction of $P^{\Pi}$ actually leads to a partition of the population into intervals \[U=\bigcup_{1\leq r\leq n}  ]k_{r-1},k_r]\] such that, if $\mathbb{S}\sim DSD(P^{\Pi})$, then:
\begin{itemize}
\item $\mathbb{S}$ has at most one point into each open interval $]k_{r-1},k_r[$, 
\item $\mathbb{S}$ has at least one and at most three points into each closed interval $[k_{r-1},k_r]$.
\item $\mathbb{S}$ has at most two points into each open interval $]k_{r-1},k_{r+1}[$,

\end{itemize} 
To help understand the way a sample is drawn, Figure \ref{ExplicationTheorem2} shows various feasible and unfeasible samples for $n=3$, giving a graphical representation of the previous properties.

\begin{figure}[h!]
\begin{center}
\caption{Examples of feasible and unfeasible samples $\mathbb{S}\sim DSD(P^{\Pi})$, $n=3$}\label{ExplicationTheorem2}
\includegraphics[scale=0.5]{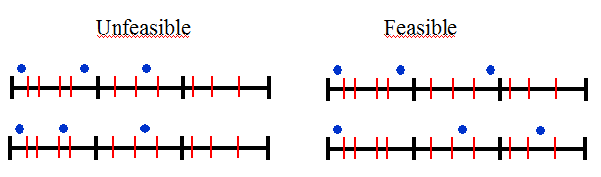}
\end{center}
\end{figure}

We finally provide an example of two matrices built by the previous method. 
\begin{example}
Let $\Pi=(\frac{1}{2},\frac{3}{4},\frac{3}{4},\frac{1}{5},\frac{2}{5},\frac{3}{5},\frac{4}{5})^{T}$ and $\Pi'=(\frac{1}{2},\frac{1}{5},\frac{3}{4},\frac{4}{5},\frac{2}{5},\frac{3}{5},\frac{3}{4})^{T}$. Observe that $\Pi'$ is a permutation of $\Pi$, and that $\Pi_{1}+\Pi_2+\Pi_3=2$. Then  

$$P^{\Pi}=\left(\begin{array}{ccccccc}
\frac{1}{2}	&	\frac{1}{2\sqrt{2}}	&	\frac{1}{2\sqrt{2}}	&	0	&	0	&	0	&	0	\\
\frac{1}{2\sqrt{2}}	&	\frac{3}{4}	&	-\frac{1}{4}	&	0	&	0	&	0	&	0	\\
\frac{1}{2\sqrt{2}}	&	-\frac{1}{4}	&	\frac{3}{4}	&	0	&	0	&	0	&	0	\\
0	&	0	&	0	&	\frac{1}{5}	&	\frac{\sqrt{2}}{5}	&	\frac{2}{5\sqrt{3}}	&	\frac{\sqrt{2}}{5\sqrt{3}}	\\
0	&	0	&	0	&	\frac{\sqrt{2}}{5}	&	\frac{2}{5}	&	\frac{2\sqrt{2}}{5\sqrt{3}}	&	\frac{2}{5\sqrt{3}}	\\
0	&	0	&	0	&	\frac{2}{5\sqrt{3}}	&	\frac{2\sqrt{2}}{5\sqrt{3}}	&	\frac{3}{5}	&	-\frac{\sqrt{2}}{5}	\\
0	&	0	&	0	&	\frac{\sqrt{2}}{5\sqrt{3}}	&	\frac{2}{5\sqrt{3}}	&	-\frac{\sqrt{2}}{5}	&	\frac{4}{5}	\\

 \end{array}\right),$$ $$P ^{\Pi'}=\left(\begin{array}{ccccccc}
\frac{1}{2}	&	\frac{1}{\sqrt{10}}	&	\frac{\sqrt{3}}{2\sqrt{14}}	&	\frac{\sqrt{3}}{\sqrt{70}}	&	\frac{1}{\sqrt{35}}	&	\frac{1}{\sqrt{65}}	&	\frac{1}{2\sqrt{26}}	\\
\frac{1}{\sqrt{10}}	&	\frac{1}{5}	&	\frac{\sqrt{3}}{2\sqrt{35}}	&	\frac{\sqrt{3}}{5\sqrt{7}}	&	\frac{\sqrt{2}}{5\sqrt{7}}	&	\frac{\sqrt{2}}{5\sqrt{13}}	&	\frac{1}{2\sqrt{65}}	\\
\frac{\sqrt{3}}{2\sqrt{14}}	&	\frac{\sqrt{3}}{2\sqrt{35}}	&	\frac{3}{4}	&	-\frac{1}{2\sqrt{5}}	&	-\frac{1}{\sqrt{30}}	&	-\frac{\sqrt{7}}{\sqrt{390}}	&	-\frac{\sqrt{7}}{4\sqrt{39}}	\\
\frac{\sqrt{3}}{\sqrt{70}}	&	\frac{\sqrt{3}}{5\sqrt{7}}	&	-\frac{1}{2\sqrt{5}}	&	\frac{4}{5}	&	-\frac{\sqrt{2}}{5\sqrt{3}}	&	-\frac{\sqrt{14}}{5\sqrt{39}}	&	-\frac{\sqrt{7}}{2\sqrt{195}}	\\
\frac{1}{\sqrt{35}}	&	\frac{\sqrt{2}}{5\sqrt{7}}	&	-\frac{1}{\sqrt{30}}	&	-\frac{\sqrt{2}}{5\sqrt{3}}	&	\frac{2}{5}	&	\frac{2\sqrt{7}}{5\sqrt{13}}	&	\frac{\sqrt{7}}{\sqrt{130}}	\\
\frac{1}{\sqrt{65}}	&	\frac{\sqrt{2}}{5\sqrt{13}}	&	-\frac{\sqrt{7}}{\sqrt{390}}	&	-\frac{\sqrt{14}}{5\sqrt{39}}	&	\frac{2\sqrt{7}}{5\sqrt{13}}	&	\frac{3}{5}	&	-\frac{1}{\sqrt{10}}	\\
\frac{1}{2\sqrt{26}}	&	\frac{1}{2\sqrt{65}}	&	-\frac{\sqrt{7}}{4\sqrt{39}}	&	-\frac{\sqrt{7}}{2\sqrt{195}}	&	\frac{\sqrt{7}}{\sqrt{130}}	&	-\frac{1}{\sqrt{10}}	&	\frac{3}{4}	\\

 \end{array}\right).$$ 
\end{example}

\section{Optimal sampling designs}\label{SecOpt}
\subsection{A generic optimization problem}
In this section, we are given a fixed vector of weights $w$, and we estimate any variable $y$ by the linear estimator of weight $w$, $\hat{t}_y=\sum_{k^\in \mathbb{S}} w_k y_k$. It is common to search for sampling designs providing \textit{representative} or \textit{balanced} samples for a set of $Q$ auxiliary variables, where a sample $S$ is representative for $x$ if $\sum_{k^\in S} w_k x_k=t_x$. \citet{deville2004efficient} provide a general method, called the cube method, for selecting approximately balanced samples with (equal or unequal) fixed inclusion probabilities (usually $\Pi_k=w_k^{-1}$, so that the linear estimator is the Horvitz-Thompson estimator), and any number of auxiliary variables. 

Regarding DSDs, our approach to provide approximately balanced samples is to interpret representativity as follows: a sampling design is representative if $MSE(\hat{t}_{x^qw})=0$ for all $1\leq q\leq Q.$ This approach is for instance considered in \citet{fuller2009some}, but the inclusion probabilities of the optimal sampling design are then unknown. As the MSE is available in a closed form for DSDs, we consider representativity as an optimization problem. Precisely, we minimize the sum of the $\mathrm{MSE}$ over the set of determinantal sampling designs with first order inclusion probabilities $\pi_k=\Pi_k,\, k=1,\ldots, N$, that is over the set \[\Theta=\{K^{\Pi}\,|\,spec(K^{\Pi}) \in \left[0,1\right]^N, K^{\Pi}_{kk}=\Pi_k,\forall 1 \leq k \leq N \}=\{K^{\Pi}\,|\,0\leq K^{\Pi}\leq I_N, diag(K^{\Pi})=\Pi\}\] of contracting matrices of diagonal $\Pi$. Since the diagonal of $K^{\Pi}$ is fixed the bias is constant, and the problem 
$\underset{\Theta}{\mathcal{M}in} \displaystyle\sum_{q=1}^Q \mathrm{MSE}(\hat{t}_{x^qw})$
admits the following formulations (we pose $z^q=w*x^q$):
\begin{problem}\label{ProbP}
Find \[\underset{\Theta}{\arg\min} \sum_{q=1}^{Q} (z^q)\,^T(I_N-K^{\Pi})*\overline{K^{\Pi}})z^q=\underset{\Theta}{\arg\max} \sum_{q=1}^{Q} (z^q)\,^T(K^{\Pi}*\overline{K^{\Pi}}
)z^q.\]
\end{problem}

We now analyze this optimization problem.
\begin{itemize}
\item Problem \ref{ProbP} is well-posed, because we consider the optimization of a continuous function on the (convex) compact set of contracting matrices. 
\item The parameter set $\Theta$ is a \textit{projected spectrahedron}, that is the projection of the intersection of the cone of positive semidefinite matrices and an affine space. The optimization problem we consider here is then a particular case of the \textit{semidefinite optimization problems}.
\item This problem is nonetheless a non-convex problem, because the objective function for the minimization problem is non-convex. Algorithmic difficulties can thus be challenging, notably when $N$ is large.
\item Consider the case of nonnegative variables $x^q$. In this case the objective function of the minimization problem is concave in $K^{\Pi}$, and the mimimum is thus attained at an extremal point of the convex set $\Theta$.
\end{itemize}

The number of studies on semidefinite optimization has been growing rapidly (in the convex and linear setting) since the 90's (see for instance\citet{blekherman2013semidefinite}, \citet{vandenberghe1996semidefinite}). But while efficient algorithms exist in the case of a strictly convex objective function, problems are extremely difficult otherwise. One of the difficulties in the case of linear or concave objective functions is that the extreme points of spectrahedra do not generally admit a simple characterization. Indeed, the problem of deciding whether a given matrix is an extreme point of a given spectrahedron is NP-hard for many spectrahedra. This is for instance the case for the \textit{elliptope of correlation matrices}. 


In practice, existing semidefinite optimization algorithms fail to produce optimal solutions when $N$ is large. Indeed, we have seen in Section \ref{Sec:ConstructingUnequal} that producing a projection element in $\Theta$ (projections are extreme points, but not all extreme points are projections) is in itself a difficult task. However, we will see in our simulation studies (Section \ref{SecSim}) that the construction of a specific DSD (the $DSD(P^{\Pi})$ of Theorem \ref{theovin}) leads to (empirical) optimal solutions for one auxiliary variable. In the following, we treat the following cases:
\begin{itemize}
\item Theoritical results for $\sum_k \Pi_k\leq 1$.
\item Algorithmic minimization results, $N\leq 40$.
\item Presentation of an empirical algorithm.
\end{itemize}

The performances of the empirical algorithm are presented in Section \ref{SecSim}, for $N=5891$.

\subsection{Minimization over sampling designs of average size (less than) one}

We first consider equal-probability determinantal sampling designs of average size one. In this case, the parameter space for Problem \ref{ProbP} is 
\begin{Large}$\displaystyle\sigma$\end{Large}$=\{0\leq K\leq I_N, K_{kk}=\frac{1}{N}\}$, the spectrahedron of positive semidefinite matrices of diagonal $\frac{1}{N}$ (this set is homothetic to the set of correlation matrices, also known as the elliptope, which is the set of positive semidefinite matrices of diagonal $1$). The literature on the elliptope and linear optimization over it is abundant, see for instance  \citet{ycart1985extreme}, \citet{grone1990extremal}, \citet{laurent1995positive}, \citet{laurent1996facial}, \citet{kurowicka2003parameterization}, \citet{laurent2014new}. It is known that (for real matrices): 

\begin{theorem}[Linear optimization over the elliptope]
$\,$\\
\begin{enumerate}
\item For any integer $k$ such that $\begin{pmatrix} k+1\\ 2 \end{pmatrix} \leq N$, there exists a matrix of rank $k$ that is an extreme point of \begin{Large}$\displaystyle\sigma$\end{Large} ( \citet{grone1990extremal} Theorem 2).
\item The \textit{vertices} of \begin{Large}$\displaystyle\sigma$\end{Large} (extreme points where the normal cone to \begin{Large}$\displaystyle\sigma$\end{Large} is of rank $N$) are the projections of \begin{Large}$\displaystyle\sigma$\end{Large} (rank $1$ matrices).
\item It is NP-hard to decide whether the optimum of linear optimization problem $\displaystyle \max_{K\in \displaystyle\sigma}\langle A,K\rangle$ is reached at a vertex.
\end{enumerate}
\end{theorem}

Otherwise stated, the minimization of a linear function over the elliptope can be considerably hard, and the solution may not be a projection matrix.

Surprisingly, for this particular set (equal probability determinantal sampling designs of average size $1$), the quadratic problem is much more simpler than the linear one. Actually, the minimization Problem \ref{ProbP} for all unequal-probability sampling designs of average size less than $1$ (not only the determinantal ones) admits a simple solution. 

\begin{theorem}[Optimal sampling design, average size less than $1$]\label{ThMinSize2}
Let $\Pi$ be a vector of inclusion probabilities such that $\sum_{k=1}^{N}\Pi_k\leq 1$, and $x^1,\cdots, x^Q$ be nonnegative variables. There exists a unique sampling design that minimize $\sum_{q=1}^{Q}MSE(\hat{t}_{x^qw})$ within all sampling designs with fixed first order inclusion probabilities $\pi_k=\Pi_k$. It is the determinantal sampling design $\mathcal{P}=DSD(K^{\Pi})$, where $K^{\Pi}$ is any rank $1$ matrix with the prescribed diagonal. This sampling design $\mathcal{P}$ consists in sampling no element with probability $1-\sum_{k=1}^{N}\Pi_k$, and the single element $k$ with probability $\Pi_k$.
\end{theorem}

\begin{proof}
Let $\cP$ be any sampling design with fixed first order inclusion probabilities $\pi_k=\Pi_k$. As for $k\neq l$, $\Delta_{kl}\geq -\pi_k\pi_l$ then \begin{eqnarray*} 
\displaystyle\sum_{q=1}^Q var(\hat{t}_{x^qw})&=& \sum_{q=1}^{Q}\left[\displaystyle\sum_{k\in U} (w_k x^q_k)^2(\Pi_k-\Pi_k^2) +\displaystyle\sum_{k\neq l\in U} w_kx^q_k w_l x^q_l \Delta_{kl}\right]\\
&\geq& \displaystyle\sum_{q=1}^{Q}\left[\sum_{k\in U} (w_k x^q_k)^2(\Pi_k-\Pi_k^2) -\displaystyle\sum_{k\neq l\in U} w_k x^q_k w_l x^q_l \pi_k\pi_l\right],
\end{eqnarray*} with equality iff for $k\neq l$, $\Delta_{kl}= -\pi_k\pi_l$ that is $\pi_{kl}=0$.
The only sampling design that satisfies these equalities is $\cP$, which is thus the optimal design.\\ 
Consider now $K^{\Pi}=bb\,^T$ a rank one matrix with the prescribed diagonal. Then $||b||^2=\sum_{k=1}^{N}\Pi_k\leq 1$, and $K^{\Pi}$ is a contraction of rank $1$. it follows that $DSD(K^{\Pi})$ exists, and has no more than $1$ element by Corollary \ref{corEspTaille}, so that $\pi_{kl}=0$, $k\neq l$. Finally $DSD(K^{\Pi})$ achieves this lower bound, and 
$DSD(K^{\Pi})=\mathcal{P}$.
\end{proof}

If $\sum_{k=1}^{N}\Pi_k= 1$ (in particular if $\Pi_k=\frac{1}{N}$) we get the following corollaries:
\begin{corollary}[Minimization over the elliptope]\label{CorMinSize1}
Assume the variables $x^1,\cdots, x^Q$ are nonnegative.
Then the solutions of Problem \ref{ProbP} over \begin{Large}$\displaystyle\sigma$\end{Large} are the rank one projections with diagonal $\frac{1}{N}$ (vertices of \begin{Large}$\displaystyle\sigma$\end{Large}).\\
More generally, the solutions of Problem \ref{ProbP} over \begin{Large}$\displaystyle\Theta$\end{Large} with $\sum_{k=1}^{N}\Pi_k= 1$ are the rank one projections with diagonal $\Pi_k$.
\end{corollary}


\begin{corollary}[SRS$(1)$ is optimal]\label{CorMinSize2}
The sampling design with equal first order inclusion probabilities $\pi_k=\frac{1}{N}$ that minimizes the sum of the MSEs for nonnegative variables is the SRS of size $1$, which is determinantal.
\end{corollary}

Nonnegativity is crucial in the previous results. Consider the following example:
\begin{example}\label{Ex1-1}
Let U=\{1,2\}, $x_1=-1, x_2=1$ and $\Pi_1=\Pi_2=\frac{1}{2}$. Then the variance of any equal-weighted sampling design of average size one that satisfies the Sen-Yates-Grundy conditions is $var(\hat{t}_{xw})=2-8\Delta_{1,2}\geq 2$, which is the variance of the estimator under Poisson sampling.
\end{example}

For more complex spectrahedra ($\sum_k \Pi_k>1$), a characterization of the solutions of the problem \ref{ProbP} is unknown. In particular, the question whether the solutions are always projections for integer sums remains open. 

\subsection{Results of the minimization algorithm, $N\leq 40$}
We performed nonlinear semidefinite optimization using specific algorithms (\citet{polyak1992modified}, (\citet{tutuncu2001sdpt3}), for various vectors of inclusion probabilities, integers $N\leq 40$ (size of the population) and auxiliary variables . And our empirical conclusions are:
\begin{itemize}
\item When $\sum_{k} \Pi_k=n$ is an integer, the minimizer is always a projection.
\item When $\Pi_k=\frac{n}{N}$ and $n$ divides $N$, and for one auxiliary variable only ($q=1$), the optimal determinantal sampling design is stratified, that is the solution matrix is (up to a change of base) a block diagonal matrix, with each block a rank one projection.
\item However, for more than one auxiliary variable ($q>1$) the solution is generally not stratified, for $\Pi_k=\frac{n}{N}$ and $n$ divides $N$.
\end{itemize}

\subsection{An empirical algorithm}

The previous results suggest that the following method will produce a low MSE/variance estimator, but for $1$ nonnegative auxiliary variable $x$ only. As before, we search for a sampling design with fixed first order inclusion probabilities $\Pi_k$ such that $\sum_{k=1}^N \Pi_k=n$, and the vector of weights $w$ is fixed.

\begin{algo}\label{algOpt}
Perform the following steps:
\begin{enumerate}
\item Rank the elements by the auxiliary variable $z=w*x$. This produces a permutation $\sigma$ on the population.
\item Construct the projection matrix with diagonal $\Pi^{\sigma}_{k}=\Pi_{\sigma(k)}, 1\leq k\leq N$ as in Theorem \ref{theovin}.
\item Sample $\mathbb{S}$ from $DSD(\Pi^{\sigma})$ on the population $\{1,\cdots, N\}$ by Algorithm \ref{alg1}.
\item Set $\hat{t}_{x^w}=\sum_{k\in \mathbb{S}}w_{\sigma(k)} x_{\sigma(k)}$.
\end{enumerate}
\end{algo}

Indeed, the algorithm will actually produce a perfect estimator for a stratified auxiliary variable (Theorem \ref{thmQDeville} and Corollary \ref{corovin}). And if the conditions of Theorem \ref{thmQDeville} are not fulfilled, then the matrix will nonetheless be quasi-stratified, and a rank $1$ contraction matrix on each stratum. Therefore, restricted to each stratum, the solution will achieve the minimal variance by Theorem \ref{ThMinSize2}. \\

We finally describe a method that improves Algorithm \ref{algOpt} by performing rotations that decrease the objective function (MSE of the estimator in our case). In Subsection \ref{SubSecEC}, we described a rotation $Q_2$ that changes two diagonal coefficients $a_1$ and $a_2$ into two new elements $\Pi_1$ and $\Pi_2$. By letting $a_1=\Pi_1$ and $a_2=\Pi_2$ in the formulas (with $\Pi_1\neq \Pi_2$), we end up with two rotation matrices. The trivial solution $I_2$ ($t=0$) and a second non-trivial solution that changes the off-diagonal elements without altering neither the diagonal nor the spectrum:     

\[Q_2^t = \left( \begin{array}{cc}
\sin \theta & \cos \theta  \\
-\cos \theta & \sin \theta \\
  \end{array}\right ), 
\text{ with }t = \frac{2Re (a_{21})}{\Pi_2-\Pi_{1}}, \sin \theta = \frac{1}{\sqrt{1+t^2}}, \cos \theta = t\sin \theta.\]  

For a given matrix $K$ with diagonal $\Pi$, this provides for any pair $(k,l)$ such that $\Pi_k\neq \Pi_l$ a rotation matrix $W(k,l)$ that changes only the elements $K_{kl}$ and $K_{lk}$. We thus derive a ``greedy'' improvement of Algorithm \ref{algOpt}.
\begin{algo}\label{algOpt2}
Perform the following steps:
\begin{enumerate}
\item Construct the matrix $\Pi^{\sigma}$ as in Algorithm \ref{algOpt}. 
\item For all $(k,l)$ such that $\Pi_k\neq \Pi_l$, update $\Pi^{\sigma}$ into $W(k,l)\Pi^{\sigma}W(k,l)^T$ iff this decreases the objective function.
\item Conclude as in Algorithm \ref{algOpt}.
\end{enumerate}
\end{algo}

\section{Simulation studies}\label{SecSim}

We present below our simulations studies. They are based on real data sets of auxiliary variables. We consider the problem of selecting primary units (PUs) for the French master sample. The population consists of $N=5891=43*137$ geographical entities partitioning the French Continental territory (figure \ref{france}). 

\begin{figure}[h!]
\begin{center}
\caption{Partitioning the French continental territory in 5891 primary units}\label{france}
 \includegraphics[scale=0.3]{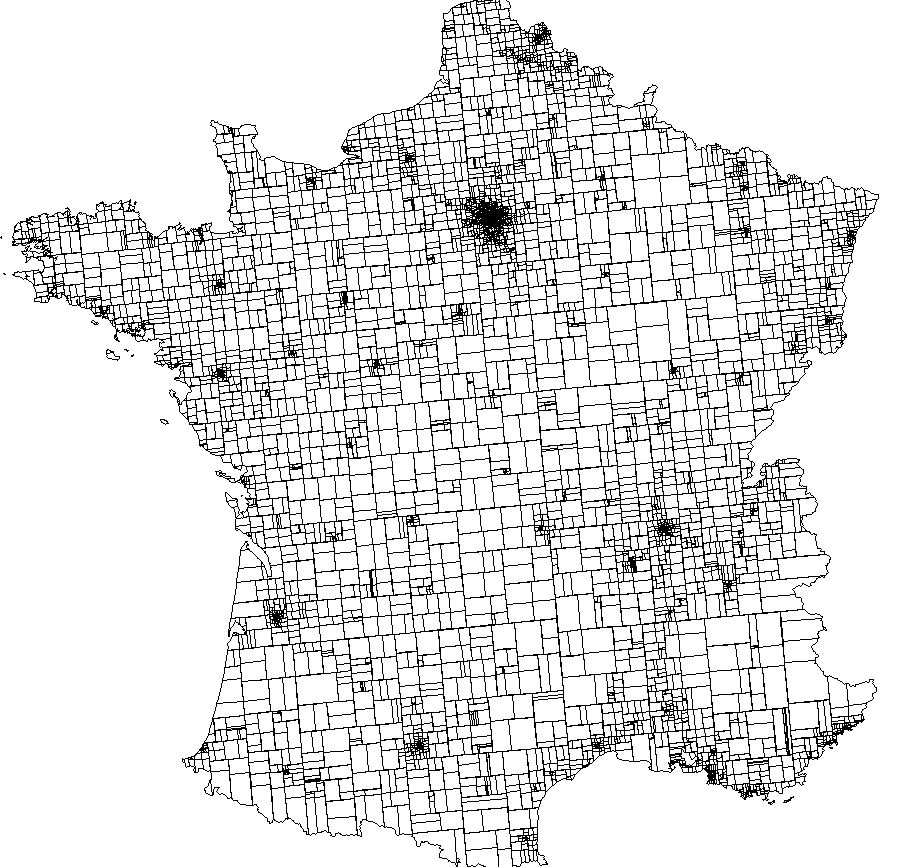}
\end{center}
\end{figure}

For each PU we observe the number of main dwellings ($x^1$), the total amount of pensions paid to the inhabitants ($x^2$), of unemployment benefit ($x^3$), of revenues from economic activity ($x^4$), normalized so that each total is $1$. Computing the variance of these normalized variables is then equivalent with computing the coefficient of variation (CV) of the original variables. We then consider two sets of inclusion probabilities:  $\Pi_k=nN^{-1}$ and $\Pi_k'=nx^1_kt_{x^1}^{-1}$ for $n$ =$30,43,100,137,200,400,600.$ \\

We perform (approximated) balanced sampling on each single variable. Precisely, for each variable $x^q$, $q=2,3,4$, each sample size and each set of inclusion probabilities we perform the first two steps of Algorithm \ref{algOpt} with the weights corresponding to those of the Horvitz-Thompson estimator: $w_k=\Pi_k^{-1}$ and $w_k'=\Pi_k^{'-1}$. We then compute the exact variance of this estimator for $x^q$ using Proposition \ref{PropAlg}.  We compare the result with the variance obtained for other popular sampling designs: a same size systematic sampling design controlled by $x^q$, a second one controlled by $x^qw$, and a sampling design resulting from the cube method balanced on $x^q$ (\citet{deville2004efficient}). For these other designs, the variance has to be computed by Monte-Carlo simulations. Systematic sampling is both simple to implement and known to achieve low variance for ranked data, and the cube method is considered as one of the best method to produce near optimal estimators. These two types of sampling are now commonly used in official statistics.\\

In each case the determinantal sampling design leads to the lowest CV (figure \ref{fig:1}). For the second set of inclusion probabilities, the systematic design has better results when controlled by $x^qw'$ rather than by just $x^q$. Moreover, for the latter, the CV might locally increase with the sample size. \\

\begin{figure}[h!] 
\caption{Empirical Studies} \label{fig:1}
\begin{subfigure}[h!]{0.48\textwidth}
\includegraphics[scale=0.42]{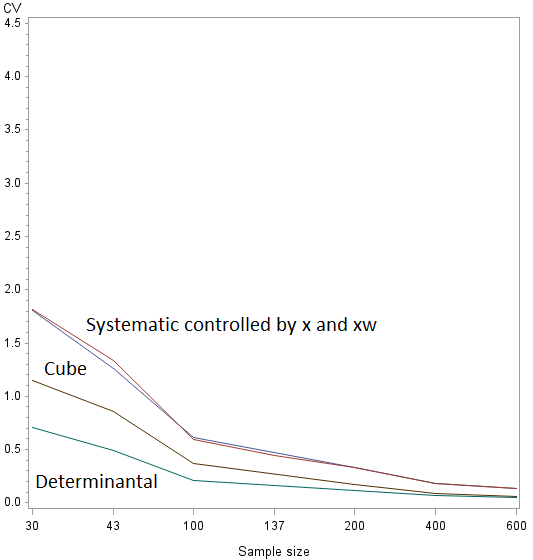}
\caption{$x^2$ and $\Pi_k=nN^{-1}$} \label{fig:a}
\end{subfigure}\hspace*{\fill}
\begin{subfigure}[h!]{0.48\textwidth}
\includegraphics[scale=0.42]{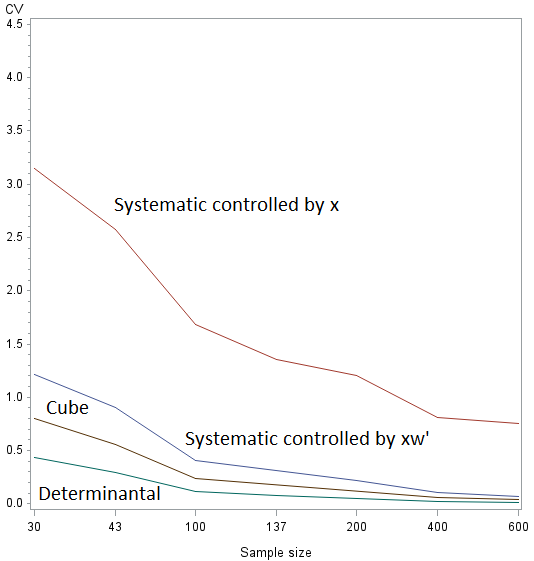}
\caption{$x^2$ and $\Pi_k'=nx^1_kt_{x^1}^{-1}$} \label{fig:b}
\end{subfigure}

\medskip
\begin{subfigure}[h!]{0.48\textwidth}
\includegraphics[scale=0.42]{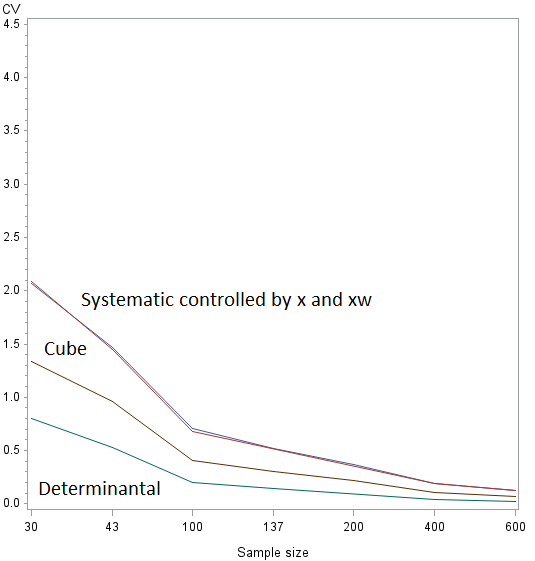}
\caption{$x^3$ and $\Pi_k=nN^{-1}$} \label{fig:c}
\end{subfigure}\hspace*{\fill}
\begin{subfigure}[h!]{0.48\textwidth}
\includegraphics[scale=0.42]{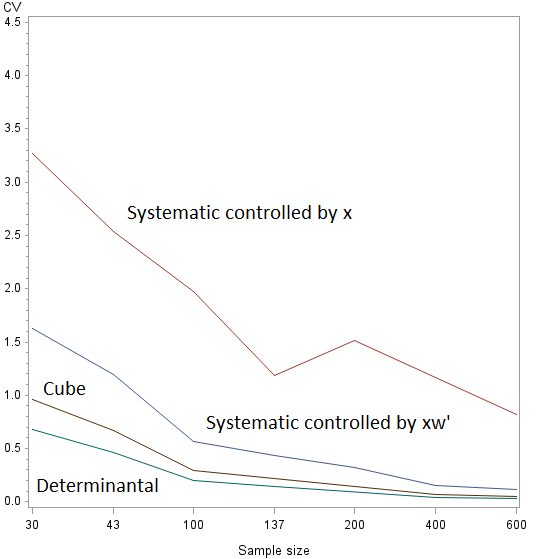}
\caption{$x^3$ and $\Pi_k'=nx^1_kt_{x^1}^{-1}$} \label{fig:d}
\end{subfigure}

\medskip
\begin{subfigure}[h!]{0.48\textwidth}
\includegraphics[scale=0.42]{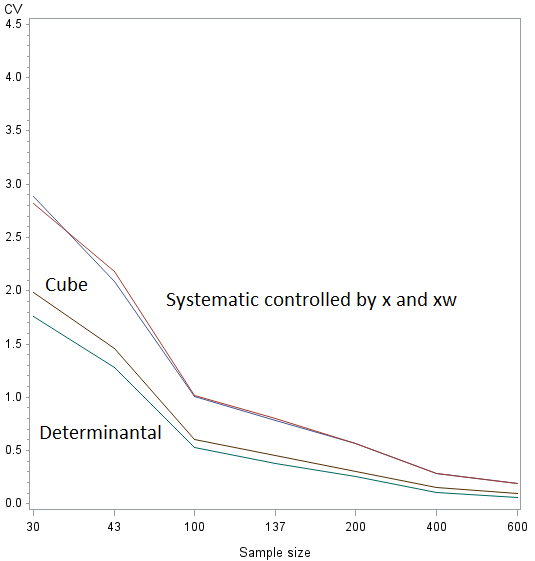}
\caption{$x^4$ and $\Pi_k=nN^{-1}$} \label{fig:e}
\end{subfigure}\hspace*{\fill}
\begin{subfigure}[h!]{0.48\textwidth}
\includegraphics[scale=0.42]{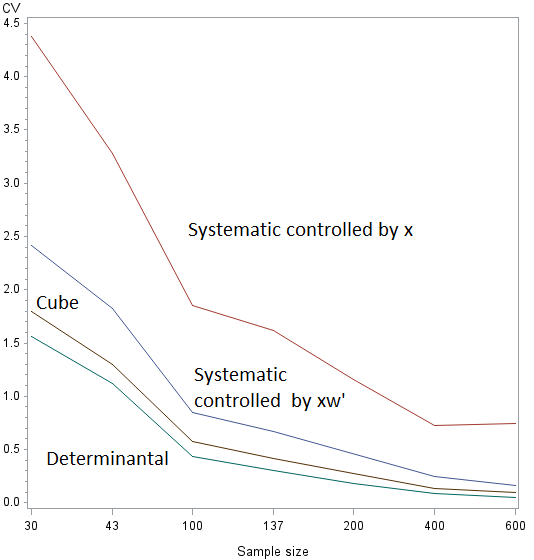}
\caption{$x^4$ and $\Pi_k'=nx^1_kt_{x^1}^{-1}$} \label{fig:f}
\end{subfigure}

\end{figure}
\newpage

\section{Acknowledgments}
We acknowledge Xu Kai from the French National School of Statistics and Economic Administration (ENSAE) for his helpful implementation of semidefinite optimization with MATLAB. We also thank Martin Brady from the Australian Bureau of Statistics (ABS) and Antoine Chambaz from Modal'X (Universit\'e Paris-Nanterre) for their useful comments. 
\section*{References}
\bibliographystyle{apalike}
\bibliography{biblio}
\end{document}